\documentclass[12pt,twoside]{article}

\usepackage{makecell}
 \usepackage{float}
\usepackage{graphicx}
\usepackage{epstopdf}
\usepackage[figuresright]{rotating}
\usepackage{graphicx}
\usepackage{epic}
\usepackage{multirow}
\usepackage{tikz}
\usepackage{xcolor}
\usepackage{threeparttable}
\usetikzlibrary{arrows,shapes,chains}

\usepackage{graphicx}
\usepackage{makecell}
\renewcommand{\paragraph}{\roman{paragraph}}
\usepackage[a4paper]{geometry}
\makeatletter
\renewcommand\title[1]{\gdef\@title{\reset@font\Large\bfseries #1}}
\renewcommand\section{\@startsection {section}{1}{\z@}%
                                   {-3.5ex \@plus -1ex \@minus -.2ex}%
                                   {2.3ex \@plus.2ex}%
                                   {\normalfont\large\bfseries}}
\renewcommand\subsection{\@startsection{subsection}{2}{\z@}%
                                     {-3ex\@plus -1ex \@minus -.2ex}%
                                     {1.5ex \@plus .2ex}%
                                     {\normalfont\normalsize\bfseries}}
\renewcommand\subsubsection{\@startsection{subsubsection}{3}{\z@}%
                                     {-2.5ex\@plus -1ex \@minus -.2ex}%
                                     {1.5ex \@plus .2ex}%
                                     {\normalfont\normalsize\bfseries}}

\def\@runningauthor{}\newcommand{\runningauthor}[1]{\def\runningauthor{#1}}
\def\@runningtitle{}\newcommand{\runningtitle}[1]{\def\runningtitle{#1}}

\renewcommand{\ps@plain}{%
\renewcommand{\@evenhead}{\footnotesize\scshape \hfill\runningauthor\hfill}
\renewcommand{\@oddhead}{\footnotesize\scshape \hfill\runningtitle\hfill}}
\newcommand{\F}{\mathbb{F}}

\newcommand {\C}{{\mathcal{C}}}

\newcommand {\SQ}{{\mathcal{SQ}}}
\newcommand {\NSQ}{{\mathcal{NSQ}}}

\makeatletter

\newcommand{\Rmnum}[1]{\expandafter\@slowromancap\romannumeral #1@}

\makeatother

\usepackage{CJK}
\usepackage{indentfirst}

\usepackage{amsmath}
\usepackage{amssymb}
\usepackage{bm}
\usepackage{amssymb}
\usepackage{amsthm}
\usepackage{multirow,booktabs}
\usepackage{cases}
\usepackage{tabularx}
\usepackage{adjustbox}
\usepackage[figuresright]{rotating}
\usepackage{longtable}
\usepackage{booktabs}
\usepackage{multirow}
\usepackage{color}
\usepackage{cite}
\usepackage{multirow,booktabs}
\usepackage{cases}
\usepackage{tabularx}
\usepackage{adjustbox}
\usepackage[figuresright]{rotating}
\usepackage{longtable}
\usepackage{booktabs}
\usepackage{multirow}
\usepackage{color}
\usepackage{threeparttable}
\makeatother

\usepackage{amsthm,amsmath,amssymb}
\usepackage{cite}
\usepackage{graphicx}

\usepackage[colorlinks=true,citecolor=black,linkcolor=black,urlcolor=blue]{hyperref}

\theoremstyle{plain}
\newtheorem{theorem}{Theorem}[section]

\newtheorem{lemma}[theorem]{Lemma}

\newtheorem{proposition}[theorem]{Proposition}

\theoremstyle{definition}
\newtheorem{definition}[theorem]{Definition}
\newtheorem{example}[theorem]{Example}

\theoremstyle{remark}
\newtheorem{remark}[theorem]{Remark}

\runningauthor{}

\date{}

\begin{document}
\begin{sloppypar}

  \title{\Large{New ternary self-orthogonal codes and related LCD codes from weakly regular plateaued functions}
  \thanks{The authors are with School of Mathematics, Hefei University of Technology, Hefei, 230601, China (email: dengchengxiemath@163.com, zhushixinmath@hfut.edu.cn, yanglimath@163.com).}
  \thanks{This research is supported by the National Natural Science Foundation of China (Grant Nos. U21A20428 and 12171134).}}
  \author{Dengcheng Xie, Shixin Zhu\thanks{Corresponding author}, Yang Li}

  \maketitle
  
  \begin{abstract}
    A linear code is said to be self-orthogonal if it is contained in its dual. 
    Self-orthogonal codes are of interest because of their important applications, such as for constructing linear complementary dual (LCD) codes and quantum codes. 
    %Self-orthogonal codes have received attention due to their practical applications in constructing linear complementary dual (LCD) codes and quantum codes.
    In this paper, we construct several new families of ternary self-orthogonal codes by employing  weakly regular plateaued functions.   
    Their parameters and weight distributions are completely determined. 
    Then we apply these self-orthogonal codes to construct several new families of ternary LCD codes. 
    As a consequence, we obtain many (almost) optimal ternary self-orthogonal codes and LCD codes.

  \end{abstract}

  \section{Introduction}
  %\subsection{Linear codes}
  %In this subsection, we introduce linear codes and two classes of linear codes with special properties: self-orthogonal codes and linear complementary dual (LCD) codes.
  Throughout this paper, let $p$ be an odd prime and $n$ be a positive integer. 
  Let $\F_{q}$ be the finite field with $q$ elements and $\F_q^*=\F_q\setminus \{0\}$, where $q=p^n$.    
  A $q$-ary {\em linear code} $\C$, denoted by $[n,k,d]$, is a $k$-dimensional linear subspace of $\F^n_{q}$ and 
  $d$ is called the {\em minimum Hamming weight} of $\C$.  
  Let the polynomial $1+A_1z+\cdots+A_nz^n$ denote the {\em weight enumerator} of $\mathcal{C}$, 
  where $A_i$ is the number of codewords with Hamming weight $i$ in $\mathcal{C}$.
  Then the sequence $(1,A_1,A_2,\cdots,A_n)$ is called the {\em weight distribution} of $\mathcal{C}$.
  If the number of $i$ such that $A_i\neq 0$ in the sequence $(1,A_1,A_2,\cdots,A_n)$ equals $t$, 
  then the linear code $\mathcal{C}$ is called a {\em $t$-weight code}.  
  The well-known {\em Sphere-packing bound} \cite{ternary self-orthogonal} on a $q$-ary $[n,k,d]$ linear code yields that  
  \begin{align}\label{eq.Sphere-packing}
      q^n\geq q^k\sum_{i=0}^{\left\lfloor (d-1)/2\right\rfloor}\binom{n}{i}(q-1)^i,  
  \end{align}
  %$$q^n\geq q^k\sum_{i=0}^{\left\lfloor (d-1)/2\right\rfloor}\binom{n}{i}(q-1)^i,$$
  where $\left\lfloor \cdot\right\rfloor$ denotes the floor function.    
  An $[n,k,d]$ linear code $\mathcal{C}$ is said to be {\em optimal} if there is no $[n,k,d']$ linear code such that $d'>d$  
  and said to be {\em almost optimal} if there is an optimal $[n,k,d+1]$ linear code. 
  
\subsection{Self-orthogonal codes and linear complementary dual codes}
For any $q$-ary $[n,k,d]$ linear code, its {\em dual code} is a linear code defined by  
$$\mathcal{C}^\bot=\left\{{\bf{c}}^{\bot}\in \F^n_{q}: \langle{\bf{c}}^{\bot},{\bf{c}}\rangle=0 \ {\rm for \ all}\ {\bf{c}}\in \mathcal{C}\right\}$$  
and $\C^{\perp}$ has parameters $[n,n-k]$, where $\langle\cdot,\cdot\rangle$ denotes the standard inner product. 
We call $\mathcal{C}$ a {\em self-orthogonal code} if $\mathcal{C} \subseteq \mathcal{C}^\bot$.  
%and a {\em dual-containing code} if $\mathcal{C}^\bot \subseteq \mathcal{C}$. 
%In the following, we give a lemma to illustrate the relationship between the self-orthogonality and divisibility of a $q$-ary linear code when $p=2,3$, which is useful to judge whether a binary or ternary code is self-orthogonal or not by the weight distribution of the code.
Self-orthogonal codes have been widely studied since coding theory was proposed \cite{SO-11}.
There are several reasons why they attracted wide interest.
On the one hand, they are closely connected with other mathematical structures such as 
combinatorial $t$-design theory \cite{t-design}, group theory \cite{C-1} and modular forms \cite{Shi-book}. 
On the other hand, they have vital applications in constructing quantum codes \cite{quantum-1,quantum-2,quantum-3,quantum-4,quantum-5}.  
Note also that there are too many self-orthogonal codes as the dimension $k$ increases and the inner product changes. 
Li $et~al.$ \cite{Li-Xu-Zhao} and Shi $et~al.$ \cite{conjectures} determined the minimum distance of an optimal binary Euclidean self-orthogonal $[n,k]$ code for $k \leq 6$.  
Bouyukliev $et~al.$ \cite{SIAM-SO} also classified ternary Euclidean self-orthogonal codes and quaternary Hermitian self-orthogonal code for small parameters. 
For more details on self-orthogonal codes, we refer to \cite{SO-40,sym-SO,Kim-SO,IEEE-kim-Lee,SO-DM-Vpless} and the references therein.

%It is well-known that the self-orthogonality of a binary or ternary linear code can be determined by its weight distribution. 
%Specifically, a binary (resp. ternary) linear code $\C$ is self-orthogonal if and only if the Hamming weight of every codeword in $\mathcal{C}$ is divisible by $4$ (resp. 3) \cite{ternary self-orthogonal}. 

A linear code $\mathcal{C}$ is said to be {\em linear complementary dual (LCD)} if $\mathcal{C}\cap\mathcal{C}^\bot=\{\textbf{0}\}$, 
where $\textbf{0}$ is a zero vector. 
With \cite{LCD-Massey}, an $[n,k]$ linear code $\mathcal{C}$ generated by a matrix $G$ is LCD if and only if $GG^T$ is nonsingular.     
It is also clear that the dual code of an LCD code is still LCD. 
At first, LCD codes were introduced by Massey \cite{LCD-Massey} in order to solve a problem in information theory.
Sendrier \cite{LCD-is-good} showed that LCD codes meet the asymptotic Gilbert-Varshamov bound.
Carlet and Guilley \cite{LCD code APP} investigated an application of binary LCD codes against Side-Channel Attacks (SCAs) and Fault Injection Attacks (FIAs). 
The study of LCD codes has thus become a hot research topic and been further carried out recently 
(see \cite{LCD(MDS)-1,LCD(MDS)-2,LCD(MDS)-3,LCD(MDS)-5,LCD(MDS)-4 LiS,LCD(MDS)-6 liY,LCD(MDS)-7,IT-Jin,MDS-LCD-equ}).
In particular, Carlet $et~al.$ showed that any $q$-ary linear code is equivalent to some Euclidean LCD code in \cite{LCD-equivalent}. 
This motivates us to study LCD codes over small finite fields.
Therefore, it is always very interesting to construct self-orthogonal codes and LCD codes in the study of coding theory.

\subsection{Related works on the constructions of linear codes}

%Recently, many authors have used general and completed functions to construct linear codes (see for example \cite{}). 
It is well-known that there are many different methods to construct linear codes. 
On the one hand, we note that one of them is the so-called generic method, which was proposed by Ding $et~al.$ in \cite{secret share 3[1]}. 
%We can summarize some outstanding recent works on this topic as follows. 
%Note that $\F_q^*$ is a  cyclic group of units of $\F_q$ which consists of $q-1$ elements. 
Specifically, we consider the {\em absolute trace function} $Tr_{p}^{n}(x)=x+x^p+x^{p^2}+\dots+x^{p^{n-1}}$ 
that maps an element $x\in \F_{q}$ to another element of $\F_p$.   
%$$Tr_{p}^{n}(x)=x+x^p+x^{p^2}+\dots+x^{p^{n-1}},$$
%where $x \in \F_q$. 
By fixing a set $D=\{d_1,d_2,\cdots,d_n\}\subseteq \F_q^*$, Ding $et~al.$ \cite{secret share 3[1]} 
proposed a general method to generate linear codes $\mathcal{C}'_D$ with the form of 
\begin{align}\label{simple C_D}
   \mathcal{C}_D'=\{\left(Tr_{p}^{n}(xd_1),Tr_{p}^{n}(xd_2),\cdots,Tr_{p}^{n}(xd_n)\right):\ x\in \F_q\}. 
\end{align}  
Furthermore, the set $D$ is called the {\em defining set} of $\mathcal{C}'_D$.
In \cite{secret share 3[1]}, Ding $et~al.$ chose the defining set $D=\{x\in \F_q^*: x^2=0\}$ and constructed several new families of $2$-weight and $3$-weight linear codes.   
Based on this generic method, many researchers have employed $p$-ary (weakly regular) bent functions and plateaued functions to provide defining sets and 
constructed infinite families of linear codes with few weights or other desired properties 
(see for example \cite{WRP31[2],RF[1],SqF[2],WRB[1],DfSeT-1,DfSeT-2,DfSeT-3,DfSeT-4,HDZ-FFA[1],Heng-DfSeT-2,Heng-multiplicative characters}).  
%Ding $et~al.$ \cite{secret share 3[1]} presented several 2-weight and 3-weight linear codes 
%and determined their weight distributions of $\mathcal{C}_D$ from defining set $D=\{x\in \F_q^*: x^2=0\}$. 
%In particular, as a generalization of bent functions and plateaued functions, Mesnager $et~al.$ \cite{WRP31[2]} also derived some good linear codes from weakly regular plateaued functions. 
%and introduced their applications on the secret sharing schemes. 

Note that compared with using only one weakly regular bent or plateaued function, 
it remains an open problem for a long time to construct linear codes with more flexible parameters by mixing two or more (possibly) different functions.  
In 2021, Cheng $et~al.$ fulfilled this gap in \cite{Cheng Y.& Cao[1]}. 
They presented a new construction of linear codes, which gives a linear code 
\begin{align}\label{Double C_D}
  \mathcal{C}_D=\{Tr_{p}^{2n}(\alpha x+\beta y)_{(x,y)\in D}:\alpha,\beta \in \F_q\},
\end{align}
where $D=\{(x,y)\in {\F_q\times \F_q}\backslash\{(0,0)\},f(x)+g(x)=c\}$ for $c \in \F_p^*$ and $f,g: \F_q\rightarrow\F_p $ are weakly regular plateaued functions. 
As a result, they obtained several classes of few-weight linear codes with good and flexible parameters by using the properties of cyclotomic fields and exponential sums.

Very recently, by employing the well-known augmented construction \cite{ternary self-orthogonal}, Heng $et~al.$ \cite{Ternary & bent functions} 
considered the augmented code of $\C'_D$ defined in Equation (\ref{simple C_D}), which has the form of
\begin{align}\label{AC_D}
  \overline{\mathcal{C}'_{D}}=\{\left(Tr_{p}^{n}(xd_1),Tr_{p}^{n}(xd_2),\cdots,Tr_{p}^{n}(xd_n)\right)+\mu{\bf{1}} : x\in \F_q, \mu \in \F_p\},
\end{align}
where the defining set $D=\{d_1,d_2,\cdots,d_n\}\subseteq \F_q^*$ and ${\bf{1}}=(1,1,\cdots,1)\in \F_p^n$. 
This construction can yield linear code with dimension increasing by $1$ if ${\bf{1}} \notin \mathcal{C}'_D$ and hence,   
Heng $et~al.$ obtained several families of ternary self-orthogonal codes of larger dimensions by choosing defining sets from bent functions. 
They also further constructed related ternary LCD codes from these self-orthogonal codes.

\subsection{Motivations and contributions} 
%According to the relationship between self-orthogonal codes and Euclidean LCD codes as well as the equivalence of Euclidean LCD codes over $\F_q$ for %$q>3$, we investigate ternary self-orthogonal codes and related ternary LCD codes. 
%Besides, it has worth to determine the weight distributions of self-orthogonal codes and find their few-weight subcodes.      
Note that weakly regular plateaued functions are generalizations of bent functions. 
Inspired by the ideas and methods in \cite{Cheng Y.& Cao[1],Ternary & bent functions}, 
a natural and interesting problem rises: \textbf{Can new infinite families of self-orthogonal codes and LCD codes are constructed via the augmented codes of 
$\mathcal{C}_D$, where the defining sets $D$ are chosen by mixing two different weakly regular plateaued functions?}  

Motivated by this problem, we study in this paper the ternary augmented code        
 \begin{align}\label{df: double AC_D}
   \overline{\mathcal{C}_{D}}=\left\{Tr_{3}^{s}(\alpha x+\beta y)_{(x,y)\in D}+\mu {\bf{1}}:  (\alpha,\beta)\in \F_{3^n}\times\F_{3^m}, \mu \in \F_3\right\},  
 \end{align}
 where $s=n+m$ with $n$ and $m$ being two positive integers and the defining set $D$ is chosen from two different weakly regular plateaued functions. 
 According to \cite{Cheng Y.& Cao[1]}, we know that $\mathbf{1}\notin \mathcal{C}_D$, and hence, the augmented construction $\overline{\mathcal{C}_{D}}$ is meaningful.   
 Specifically, for $\lambda \in \F_3^*$, we will consider the following two new defining sets:  
\begin{align}\label{eq.D_f D_g}
  &D_{fg}(0)=\{(x,y)\in \F_{3^n}\times\F_{3^m}: f(x)+g(y)+\lambda=0\}~{\rm and}\\
  &D_{g}(0)=\left\{(x,y)\in \F_{3^n}\times\F_{3^m}:Tr_{3}^{n}(x)+g(y)+\lambda=0\right\}, 
 \end{align}
 where $f(x)$ and $g(y)$ are respectively weakly regular $k_f$-plateaued and $k_g$-plateaued unbalanced functions, 
 $0\leq k_f\leq n$ and $0\leq k_g\leq m$.  
%\begin{align}
%  D_{fgh}=\left\{(x,y,z)\in {\F_q^3}^*: f(x)+g(x)+h(x)=c\right\},c\in \F_p^*,
%\end{align} 
%where ${\F_q^3}^*=\F_q^3\backslash\{(0,0,0)\}$, $f$, $g$ and $h$ are weakly regular $k_f$-plateaued, $k_g$-plateaued and $k_h$-plateaued balanced functions, respectively.
 
%For the second construction, we replace $f(x)$ with absolute trace function $Tr_{p}^{n}:\F_{p^n} \to \F_p$ and   
%choose three defining sets as follows. 
%\begin{align}
%  D_{g}(0)&=\left\{(x,y)\in \F_{p^n}\times\F_{p^m}:Tr_{p}^{n}(x)+g(y)+\lambda=0\right\},\\ 
%  D_{g}(\SQ)&=\left\{(x,y)\in \F_{p^n}\times\F_{p^m}:Tr_{p}^{n}(x)+g(y)+\lambda\in \SQ \right\},\\ 
%  D_{g}(\NSQ)&=\left\{(x,y)\in \F_{p^n}\times\F_{p^m}:Tr_{p}^{n}(x)+g(y)+\lambda\in \NSQ\right\},  
%\end{align} 
%where $\lambda \in \F_p^*$.
%Then in the following, we construct the augmented codes defined in Equation (\ref{AC_D}) with these defining sets: 
%\begin{align}\label{df:C_Dg}
%  \overline{\mathcal{C}_{D_{g}(0)}}&=\left\{Tr_{p}^{s}(\alpha x+\beta y)_{(x,y)\in D_{g}(0)}+\mu {\bf{1}}:  (\alpha,\beta)\in \F_{p^n}\times\F_{p^m}, \mu \in \F_p\right\},\\
%  \overline{\mathcal{C}_{D_{g}(\SQ)}}&=\left\{Tr_{p}^{s}(\alpha x+\beta y)_{(x,y)\in D_{g}(\SQ)}+\mu {\bf{1}}:  (\alpha,\beta)\in \F_{p^n}\times\F_{p^m}, \mu \in \F_p\right\},\\
%  \overline{\mathcal{C}_{D_{g}(\NSQ)}}&=\left\{Tr_{p}^{s}(\alpha x+\beta y)_{(x,y)\in D_{g}(\NSQ)}+\mu {\bf{1}}:  (\alpha,\beta)\in \F_{p^n}\times\F_{p^m}, \mu \in \F_p\right\}. 
%\end{align}
Moreover, our main contributions can be summarized as follows:  
\begin{enumerate}
  \item [\rm (1)] Based on these two new defining sets, we obtain several new infinite families of ternary self-orthogonal codes in 
  Theorems \ref{Th. double even ternary}, \ref{Th. double odd ternary}, \ref{Th.simple even 0,SQ,NSQ} and \ref{Th.simple odd 0,SQ,NSQ}, respectively. 
  We completely determine their weight distributions and list respectively the weight distributions in Tables \ref{tab: double even }-\ref{tab:simple odd 0,SQ,NSQ}. 
  We also determine the whole parameters of the dual codes of these self-orthogonal codes in Theorems \ref{Th.double dual} and \ref{Th.simple dual}.  
  As explicit examples, we list some (almost) optimal ternary self-orthogonal codes in Examples \ref{1,4+0+-1, even} and \ref{1,3+0+-1, odd}.
  
  \item [\rm (2)] Based on the ternary self-orthogonal codes given above, we further construct several new families of related ternary LCD codes. 
  We completely determine the whole parameters of these LCD codes as well as their dual codes in Theorems \ref{Th. LCD even ternary}-\ref{Th. LCD simple even ternary}. 
  It is worth noting that they contain several new infinite families of ternary LCD codes, which are at least almost optimal 
  with respect to the Sphere-packing bound given in Equation (\ref{eq.Sphere-packing}).  
\end{enumerate}

The paper is organized as follows. 
In Section \ref{sec2}, we review some useful basic knowledge on cyclotomic fields, weakly regular plateaued functions and the Pless power moments. 
In Section \ref{sec.auxiliary results}, we give some auxiliary results for later use. 
In Section \ref{sec.SO}, we construct several new infinite families of ternary self-orthogonal codes.  
In Section \ref{sec.LCD codes}, we investigate the dual codes of these ternary self-orthogonal codes. 
In Section \ref{sec.LCD codes}, we consider an application of ternary self-orthogonal codes in ternary LCD codes and 
derive several new infinite families of ternary LCD codes that contain (almost) optimal ternary LCD codes. 
In Section \ref{conclude}, we conclude this paper.

\section{Preliminaries}\label{sec2}
%In this section, we state some basic notations, primary mentions on cyclotomic field, introduction of weakly regular plateaued functions and the exponential sums from them will be sign significant to prove main results of this paper.  
In this section, we recall some basic knowledge on cyclotomic fields, weakly regular plateaued functions and Pless power moments.

\subsection{Cyclotomic fields}\label{sec2.1 Basic Note} 
Let $p$ be an odd prime and $\zeta_p$ be a primitive $p$-th complex root of unity. 
Let $\SQ$ and $\NSQ$ denote the set of all nonzero squares and nonsquares in $\F_p$, respectively. 
%Let $\eta$ and ${\eta_0}$ respectively be the quadratic characters of $\F_{p^n}^*$ and $\F_p^*$.
Let ${\eta_0}$ be the quadratic characters of $\F_p^*$ and $p^*={\eta_0}(-1)p=(-1)^{\frac{p-1}{2}}p$. 
%First, we review some results on cyclotomic fields and  Krawtchouk polynomials, which are important for the sequel. 
The following lemma are deduced from \cite{quadratic character[2]}. 
\begin{lemma}[\cite{quadratic character[2]}]\label{characters} 
  %Let $q=p^h$ be a prime power, where $h$ is even. 
  %Let $\eta$ and ${\eta_0}$ denote, respectively, the quadratic characters of $\F_{p^n}^*$ and $\F_p^*$.
  %Let $p^*$denotes ${\eta_0}(-1)p=(-1)^{\frac{p-1}{2}}p$. 
  Let notations be the same as above. Then we have  
  %Then we have 
  \begin{enumerate}    
      \item [\rm (1)] $\sum_{\kappa    \in \F_p^*}\eta_0(\kappa   )=0$;  
      
      \item [\rm (2)] $\sum_{\kappa   \in \F_p^*}\eta_0(\kappa   )\zeta_p^\kappa  =\sqrt{p^*}$; 
  
      \item [\rm (3)] $\sum_{\kappa   \in \F_p^*}\zeta_p^{\kappa   \tau }=1$ and 
      $\sum_{\kappa   \in \F_p}\zeta_p^{{\kappa ^2}\tau}=\eta_0(\tau)\sqrt{p^*}$ for any $\tau  \in \F_p^*.$  
  \end{enumerate}
\end{lemma}

A cyclotomic field $\mathbb{Q}(\zeta_p)$ is obtained from the rational field $\mathbb{Q}$ by adjoining $\zeta_p$. 
The field extension $\mathbb{Q}(\zeta_p)/\mathbb{Q}$ is Galois of degree $p-1$ and the Galois group of 
$\mathbb{Q}(\zeta_p)$ over $\mathbb{Q}$ is 
\begin{align*}
  {\rm Gal}(\mathbb{Q}(\zeta_p)/\mathbb{Q})=\{\sigma_\kappa  : \kappa  \in {\F_p^*}\}, 
\end{align*}
%$$Gal(\mathbb{Q}(\zeta_p)/\mathbb{Q})=\{\sigma_\lambda  : \lambda  \in {\F_p^*}\},$$ 
where $\sigma_\kappa$ is an automorphism  of $\mathbb{Q}(\zeta_p)$ defined by $\sigma_\kappa   (\zeta_p)=\zeta_p^\kappa $. 
For any $\kappa   \in \F_p^*$ and $\tau   \in \F_p$, 
we clearly have $\sigma_\kappa  (\zeta_p^\tau  )=\zeta_p^{\kappa  \tau  }$ and it follows from Lemma \ref{characters} that 
$\sigma_\kappa   (\sqrt{p^*}^n)=\eta_0^n(\kappa  )\sqrt{p^*}^n$. 
Hence, the cyclotomic field $\mathbb{Q}(\zeta_p)$ has a unique quadratic subfield $\mathbb{Q}(\sqrt{p^*})$ 
and 
\begin{align*}
  {\rm Gal}\left(\mathbb{Q}\left(\sqrt{p^*}\right)/\mathbb{Q}\right)=\{1,\sigma_\mu \}\ {\rm for}\ \mu  \in \NSQ.  
\end{align*}
%$$Gal(\mathbb{Q}(\sqrt{p^*})/\mathbb{Q})=\{1,\sigma_\mu  \}$ for $\mu  \in \NSQ.$$ 

\subsection{Weakly regular plateaued functions}\label{sec2.3 WRP}

%$$Tr_p^n(x)=x+x^p+x^{p^2}+\dots+x^{p^{n-1}}.$$ 
\begin{definition}\label{def.WTrans}
  Assume that $f:\F_{q}\to \F_p$ be a $p$-ary function and $\alpha \in \F_{q}$, then the {\em Walsh transform} of a $p$-ary function $f$ is given by
  \begin{align}\label{WTrans}
    \widetilde{\mathcal{R}}_f(\alpha)=\sum_{x\in \F_{p^n}}\zeta_p^{f(x)-Tr_{p}^{n}(\alpha x)}, 
  \end{align}  
where $\zeta_p$ is a primitive $p$-th complex root of unity.
\end{definition}   
A $p$-ary function $f$ is called {\em balanced} if $\widetilde{\mathcal{R}}_f(0)=0$ 
and {\em bent} if $\vert\widetilde{\mathcal{R}}_f(\alpha)\vert=p^{\frac{n}{2}}$ for every $\alpha \in \F_{q}$. 
A bent function $f$ is called {\em regular} if $\widetilde{\mathcal{R}}_f(\alpha)=\sqrt{p}^n\zeta_p^{f^*(\alpha)}$ 
and called {\em weakly regular} if there is a complex root of unity $\epsilon$ such that 
$\widetilde{\mathcal{R}}_f(\alpha)=\epsilon\sqrt{p}^n\zeta_p^{f^*(\alpha)}$ for all $\alpha \in \F_{q}$, 
where $f^*(x)$ is a $q$-ary function from $\F_{q}$ to $\F_p$ and called the {\em dual} of $f(x)$. 
Moreover, $f^*(x)$ is also  a weakly regular bent function.

As generalizations of bent functions, plateaued functions were introduced over finite fields with characteristic two in \cite{define WRF[2]}. 
Specifically, a $p$-ary function $f: \F_{q} \to \F_p$ is called {\em $k_f$-plateaued} if $|\widetilde{\mathcal{R}}_f(\alpha)|^2\in\{0,p^{n+k_f}\}$ 
for every $\alpha \in \F_{q}$, where $k_f$ is an integer satisfying $0\leq k_f \leq n$. 
It is obvious that every bent function coincides with a $0$-plateaued function. 
The {\em Walsh support} of a $k_f$-plateaued function $f$ is given by
\begin{align}\label{WS}
  \widetilde{\mathcal{S}\mathcal{R}}_f=\{\alpha \in \F_{q}:|\widetilde{\mathcal{R}}_f(\alpha)|^2=p^{n+k_f}\}.
\end{align}  
%\[\begin{split}
  %  \widetilde{\mathcal{S}\mathcal{R}}_f=\{\alpha \in \F_{p^n}:|\widetilde{\mathcal{R}}_f(\alpha)|^2=p^{n+k}\}.
  %\end{split}\] 
Clearly, $|\widetilde{\mathcal{S}\mathcal{R}}_f|$=$p^{n-k_f}$, which checks the following proposition. 

\begin{proposition}\label{dim}
  Let $f$ be a $k_f$-plateaued function from $\F_{q}$ to $\F_p$, where $k_f$ is an integer satisfying $0\leq k_f\leq n$. 
  Then there exist $p^{n-k_f}$ (resp. $p^n-p^{n-k_f}$) different $\alpha\in \F_q$ such that 
  $|\widetilde{\mathcal{R}}_f(\alpha )|^2=p^{n+k_f}$ (resp. $0$). 
  %for all $\alpha \in \F_{p^n}$, 
  %$|\widetilde{\mathcal{R}}_f(\alpha )|^2$ takes the value $p^{n+s}$ (resp. $0$) $p^{n-s}$(resp. $p^n-p^{n-s}$) times.
\end{proposition}

\begin{definition}\label{exam model}
  Let $f$ be a $k_f$-plateaued function from $\F_{q}$ to $\F_p$, where $k_f$ is an integer satisfying $0\leq k_f \leq n$. 
  Then $f$ is said to be {\em weakly regular $k_f$-plateaued} if there exists a complex root of unity $u$ such that
  %\[\begin{split}
   \begin{align}
    \widetilde{\mathcal{R}}_f(\alpha )\in \{0,up^{\frac{n+k_f}{2}}\zeta_p^{f^*(\alpha)}\}    
   \end{align}
  %\end{split}\]
for all $\alpha\in \F_{q}$, where $f^*$ is a $p$-ary function over $\F_{q}$ with $f^*(\alpha)=0$ for all $\alpha \in \F_{q}\setminus \widetilde{\mathcal{S}\mathcal{R}}_f$.  
\end{definition}

Let $\varepsilon_f=\pm 1$ be a sign related to the Walsh transform of a $p$-ary function $f$. 
Mesnager $et~al.$ \cite{WRP31[2]} and Sinak $et~al.$ \cite{Exponential sums[3]} presented the following result. 

\begin{proposition}[\cite{WRP31[2],Exponential sums[3]}]\label{prop:gcd}
Let $f:\F_{q} \to \F_p$ be a weakly regular $k_f$-plateaued function, where $k_f$ is an integer with $0\leq k_f\leq n$.    
Let $\mathcal{WRP}$ be the non-trivial class of weakly regular unbalanced $k_f$-plateaued functions $f$ satisfying the following homogeneous conditions: 
\begin{enumerate}    
  \item [\rm (1)] $f(0)=0$;
  \item [\rm (2)] $f(ax)=a^{h_f}f(x)$ for every $a\in \F_p^*$ and $x\in \F_{q}$, where $h_f$ is an even positive integer such that ${\rm gcd}(h_f-1,p-1)=1$. 
\end{enumerate} 
Then if $f \in \mathcal{WRP}$ with $\widetilde{\mathcal{R}}_f(\alpha)=\varepsilon_g{\sqrt{p^*}}^{n+k_f}\zeta_p^{f^*(\alpha)}$ for every $\alpha \in \widetilde{\mathcal{S}\mathcal{R}}_g$,
there exists an even positive integer $l_f$ such that $f^*(b\alpha)=b^{l_f}f^*(\alpha)$ for any $b\in \F_p^*$ and $\alpha \in \widetilde{\mathcal{S}\mathcal{R}}_g$,
where ${\rm gcd}(l_f-1,p-1)=1$.   
\end{proposition}

\subsection{Pless power moments}
For any $p$-ary $[n,k,d]$ linear code, 
denote by $(1,A_1,A_2,\cdots,A_n)$ and $(1,A_1^\bot,A_2^\bot,\cdots,A_n^\bot)$ the weight distribution of $\mathcal{C}$ and $\mathcal{C}^\bot$, respectively. 
The following well-knowledge results are called {\em the first four Pless power moments} \cite{ternary self-orthogonal}: 
\begin{align*}
  P_1:\ \ \sum_{j = 0}^{n}A_j=&q^k,\\   
  P_2:\ \sum_{j = 0}^{n}jA_j=&q^k(qn-n-A_1^\bot),\\
  P_3: \sum_{j = 0}^{n}j^2A_j=&q^{k-2}\left((q-1)n(qn-n+1)-(2qn-q-2n+2)A_1^\bot+2A_2^\bot\right),\\
  P_4: \sum_{j = 0}^{n}j^3A_j=&q^{k-3}[(q-1)n(q^2n^2-2qn^2+3qn-q+n^2-3n+2)\\
                        &-(3q^2n^2-3q^2n-6qn^2+12qn+q^2-6q+3n^2-9n+6)A_1^\bot\\
                        &+6(qn-q-n+2)A_2^\bot-6A_3^\bot].\\    
\end{align*} 
It should be emphasized that these four Pless power moments play important roles in determining the minimum weight of the dual code of a linear code.

% \section{Punctured codes}
%Let $f,g,h\in \mathcal{WRPB}$ with $l_f=l_g=l_h$.
%It is obvious from Lemma \ref{gcd} that $f(ax)+g(ay)+h(az)=a^{l_f}\left(f(x)+g(y)+h(z)\right)$
%for any $(x, y, z)\in \F_q^*$ and $a \in \F_p^*$.
%Hence, we can select a subset $\bar{D}_{fgh}$ of the defining set $D_{fgh}$ as Equation () 
%such that $D_{fgh}=\bigcup_{a\in \F_p^*}a\bar{D}_{fgh}$ is a partition of $D_{fgh}$,
%where 
%\begin{align}
%  D_{fgh}=\F_p^*\bar{D}_{fgh}=\{a\bar{d}: a\in \F_p^*\ {\rm and } \ \bar{d} \in \bar{D}_{fgh}\},
% \end{align}
%where we have $\frac{\bar{d}_1}{\bar{d}_2}\notin \F_p^*$ for any pair of elements $\bar{d}_1,\bar{d}_1\in \bar{D}_{fgh}$.
%Hence, the code $\mathcal{C}_{D_{fgh}}$  defined as Equation (\ref{df:C_D}) can be punctured into a shortor code $\mathcal{C}_{\bar{D}_{fgh}}$ with same dimension
%where its Hamming weights and length are given by dividing those of $\mathcal{C}_{D_{fgh}}$ with $p-1$.

%We define two sets as $$ D_{fgh,SQ}=\left\{(x,y,z)\in {\F_q^3}^*: f(x)+g(y)+h(z)\in SQ\right\}$$
%and $$ D_{fgh,NSQ}=\left\{(x,y,z)\in {\F_q^3}^*: f(x)+g(y)+h(z)\in NSQ\right\},$$
%where $f,g,h\in \mathcal{WRPB}$ and $l_f=l_g=l_h$.

%Let $n_{fgh}(SQ)$ and $n_{fgh}(NSQ)$ denote the codes $\mathcal{C}_{{D}_{fgh,SQ}}$ from set $D_{fgh,SQ}$ 
%and $\mathcal{C}_{{D}_{fgh,NSQ}}$ from set $D_{fgh,NSQ}$, respectively.  

\section{Some auxiliary results}\label{sec.auxiliary results}
From this section on, we let $p=3$. Then $\sqrt{p^*}=\sqrt{-3}$ and $\eta_0(-1)=-1$. 
Note that we still use the notations $\sqrt{3^*}$ and $\eta_0(-1)$ in the sequel for the uniform representations. 
In the following, we present some auxiliary results, which will be used for determining the parameters and weight distributions  
of ternary linear codes in the next section.  

\begin{lemma}\label{P_{f*,g*+c=0}}
  Let $f,g\in \mathcal{WRP}$ with $\widetilde{\mathcal{R}}_f(\alpha )=\varepsilon_f{\sqrt{3^*}}^{n+k_f}\zeta_3^{f^*(\alpha)}$ for every $\alpha \in \widetilde{\mathcal{S}\mathcal{R}}_f$
  and $\widetilde{\mathcal{R}}_g(\beta)=\varepsilon_g{\sqrt{3^*}}^{m+k_g}\zeta_3^{g^*(\beta)}$ for every $\beta \in \widetilde{\mathcal{S}\mathcal{R}}_g$,
   where $\varepsilon_f,\varepsilon_g\in \{ -1, 1\}$, $0\leq k_f\leq n$ and $0\leq k_g\leq m$.
  Let $s=n+m$, $\lambda\in \F_3$ and 
  \begin{align*}
      P_{f^*g^*}(0)=\#\{(\alpha,\beta) \in \widetilde{\mathcal{S}\mathcal{R}}_f\times\widetilde{\mathcal{S}\mathcal{R}}_g:f^*(\alpha)+g^*(\beta)+\lambda=0\}.    
  \end{align*} 
  Then the following statements hold. 
\begin{enumerate}
  \item [\rm (1)] If $s+k_f+k_g$ is even, then 
  $$ P_{f^*g^*}(0)=\left\{
        \begin{aligned}
        &3^{s-k_f-k_g-1}+2(-1)^{s+1}\varepsilon_f\varepsilon_g{(-3)}^{\frac{s-k_f-k_g-2}{2}}   && \lambda=0, \\    
        &3^{s-k_f-k_g-1}-(-1)^{s+1}\varepsilon_f\varepsilon_g{(-3)}^{\frac{s-k_f-k_g-2}{2}}   && \lambda\in \F_3^*.    
        \end{aligned}
        \right.
        $$
  \item [\rm (2)] If $s+k_f+k_g$ is odd, then 
  $$ P_{f^*g^*}(0)=\left\{
    \begin{aligned}
    &3^{s-k_f-k_g-1}   && \lambda=0, \\
    &3^{s-k_f-k_g-1}+(-1)^{s+1}\varepsilon_f\varepsilon_g{(-3)}^{\frac{s-k_f-k_g-1}{2}}   && \lambda=1,\\
    &3^{s-k_f-k_g-1}-(-1)^{s+1}\varepsilon_f\varepsilon_g{(-3)}^{\frac{s-k_f-k_g-1}{2}}   && \lambda=2.
  \end{aligned} 
\right.
$$ 
\end{enumerate}
\end{lemma}
\begin{proof}
  By the definition of $P_{f^*g^*}(0)$, we have
    \[\begin{split} 
      P_{f^*g^*}(0) &=\frac{1}{3}\sum_{t\in \F_3}\sum_{\alpha\in \widetilde{\mathcal{S}\mathcal{R}}_f}\sum_{\beta\in \widetilde{\mathcal{S}\mathcal{R}}_g}
      \zeta_3^{t(f^*(\alpha)+g^*(\beta)+\lambda)}\\    
           & =3^{s-k_f-k_g-1}+\frac{1}{3}\sum_{t\in \F_3^*}\zeta_3^{\lambda t}\sigma_{t}\left(\sum_{\alpha \in \widetilde{\mathcal{S}\mathcal{R}}_f}\zeta_3^{f^*(\alpha)} 
           \sum_{\beta \in \widetilde{\mathcal{S}\mathcal{R}}_g}\zeta_3^{g^*(\beta)}\right)\\     
    & =3^{s-k_f-k_g-1}+(-1)^{s+1}\varepsilon_f\varepsilon_g{\sqrt{3^*}}^{s-k_f-k_g-2}\sum_{t\in \F_3^*}\eta_0^{s-k_f-k_g}(t)\zeta_3^{\lambda t}.
          \end{split}\]
    Then the desired results follow straightforward from Lemma \ref{characters}. 
 \end{proof}

%\begin{lemma}\cite{}\label{P_{f*,g*/t}}
%  Let $f,g\in \mathcal{WRPB}$ or $f,g\in \mathcal{WRP}$ with $\widetilde{\mathcal{R}}_f(\alpha )=\varepsilon_f{\sqrt{p^*}}^{n+k_f}\zeta_p^{f^*(\alpha)}$
%  and $\widetilde{\mathcal{R}}_g(\beta)=\varepsilon_g{\sqrt{p^*}}^{n+k_g}\zeta_p^{g^*(\beta)}$ 
%  for every $\alpha \in \widetilde{\mathcal{S}\mathcal{R}}_f$ and every $\beta \in \widetilde{\mathcal{S}\mathcal{R}}_g$, respectively.
%  For $c\in \F_p^*$, define 
%  \begin{align*}
%    D_{f^*g^*,SQ}(t)=\{(\alpha,\beta) \in \widetilde{\mathcal{S}\mathcal{R}}_f\times\widetilde{\mathcal{S}\mathcal{R}}_g:\frac{f^*(\alpha)+g^*(\beta)}{c}\in SQ\},\\    
%    D_{f^*g^*,NSQ}(t)=\{(\alpha,\beta) \in \widetilde{\mathcal{S}\mathcal{R}}_f\times\widetilde{\mathcal{S}\mathcal{R}}_g:\frac{f^*(\alpha)+g^*(\beta)}{c}\in NSQ\}.    
%\end{align*} 
%  Then the following statements hold. 
%\begin{enumerate}
%  \item [\rm (1)] If $k_f+k_g$ is even, then 
%  $$ |D_{f^*g^*,SQ}(t)|=|D_{f^*g^*,NSQ}(t)|=
%        \begin{aligned}
%        &\frac{p-1}{2}\left(p^{s-k_f-k_g-1}-p^{-1}\varepsilon_f\varepsilon_g{\sqrt{p^*}}^{s-k_f-k_g}\right).    
%        \end{aligned}
%        $$
%  \item [\rm (2)] If $k_f+k_g$ is odd, then 
%  \begin{align*}
%    |D_{f^*g^*,SQ}(t)|=\frac{p-1}{2}\left(p^{s-k_f-k_g-1}+\varepsilon_f\varepsilon_g\eta_0(c) {\sqrt{p^*}}^{s-k_f-k_g}\right),\\    
%  |D_{f^*g^*,NSQ}(t)|=\frac{p-1}{2}\left(p^{s-k_f-k_g-1}-\varepsilon_f\varepsilon_g\eta_0(c){\sqrt{p^*}}^{s-k_f-k_g}\right).    
%\end{align*} 
%\end{enumerate}
%\end{lemma}

\begin{lemma}\label{P_{f*+g*+c}}
  Let $f,g\in \mathcal{WRP}$ with $\widetilde{\mathcal{R}}_f(\alpha )=\varepsilon_f{\sqrt{3^*}}^{n+k_f}\zeta_3^{f^*(\alpha)}$ for every $\alpha \in \widetilde{\mathcal{S}\mathcal{R}}_f$
  and $\widetilde{\mathcal{R}}_g(\beta)=\varepsilon_g{\sqrt{3^*}}^{m+k_g}\zeta_3^{g^*(\beta)}$ for every $\beta \in \widetilde{\mathcal{S}\mathcal{R}}_g$,
   where $\varepsilon_f,\varepsilon_g\in \{ -1, 1\}$, $0\leq k_f\leq n$ and $0\leq k_g\leq m$.
  Let $s=n+m$,  $\lambda\in \F_3^*$,   
  \begin{align*}
    P_{f^*g^*}(1)&=\#\{(\alpha,\beta) \in \widetilde{\mathcal{S}\mathcal{R}}_f\times\widetilde{\mathcal{S}\mathcal{R}}_g:f^*(\alpha)+g^*(\beta)+\lambda=1\}~{\rm and}\\    
    P_{f^*g^*}(2)&=\#\{(\alpha,\beta) \in \widetilde{\mathcal{S}\mathcal{R}}_f\times\widetilde{\mathcal{S}\mathcal{R}}_g:f^*(\alpha)+g^*(\beta)+\lambda=2\}.    
\end{align*} 
  Then the following statements hold. 
\begin{enumerate}
  \item [\rm (1)] If $s+k_f+k_g$ is even, then 
  \begin{align*}
    & P_{f^*g^*}(1)=3^{s-k_f-k_g-1}+\frac{1+3\eta_0(\lambda)}{6}(-1)^{s}\varepsilon_f\varepsilon_g{(-3)}^{\frac{s-k_f-k_g}{2}}~{\rm and}\\    
    & P_{f^*g^*}(2)=3^{s-k_f-k_g-1}+\frac{1-3\eta_0(\lambda)}{6}(-1)^{s}\varepsilon_f\varepsilon_g{(-3)}^{\frac{s-k_f-k_g}{2}} .    
\end{align*} 
  \item [\rm (2)] If $s+k_f+k_g$ is odd, then 
  \begin{align*}
    & P_{f^*g^*}(1)=3^{s-k_f-k_g-1}-\frac{1+\eta_0(-\lambda)}{2}(-1)^{s}\varepsilon_f\varepsilon_g{(-3)}^{\frac{s-k_f-k_g-1}{2}}~{\rm and}\\    
    & P_{f^*g^*}(2)=3^{s-k_f-k_g-1}+\frac{1-\eta_0(-\lambda)}{2}(-1)^{s}\varepsilon_f\varepsilon_g{(-3)}^{\frac{s-k_f-k_g-1}{2}} .    
\end{align*} 
\end{enumerate}
\end{lemma}
\begin{proof} 
    Similar to Lemma \ref{P_{f*,g*+c=0}}, we still denote  
  \begin{align*}
    P_{f^*g^*}(0)=\#\left\{(\alpha,\beta) \in \widetilde{\mathcal{S}\mathcal{R}}_f\times\widetilde{\mathcal{S}\mathcal{R}}_g:f^*(\alpha)+g^*(\beta)+\lambda=0\right\}.
\end{align*}
  Consider the exponential sum
  \[\begin{split}
    P & =\sum_{t\in \F_3}\sum_{\alpha\in \widetilde{\mathcal{S}\mathcal{R}}_f}\sum_{\beta\in \widetilde{\mathcal{S}\mathcal{R}}_g}
    \zeta_3^{t^2(f^*(\alpha)+g^*(\beta)+\lambda)}\\
    & =\sum_{\alpha\in \widetilde{\mathcal{S}\mathcal{R}}_f}\sum_{\beta\in \widetilde{\mathcal{S}\mathcal{R}}_g}
    \left(\sum_{t\in \F_3^*}\zeta_3^{t^2(f^*(\alpha)+g^*(\beta)+\lambda)}+1\right)\\ 
    & = 3^{s-k_f-k_g}+\sum_{t\in \F_3^*}\zeta_3^{\lambda t^2}\sigma_{t^2}\left(\sum_{\alpha \in \widetilde{\mathcal{S}\mathcal{R}}_f}\zeta_3^{f^*(\alpha)} 
    \sum_{\beta \in \widetilde{\mathcal{S}\mathcal{R}}_g}\zeta_3^{g^*(\beta)}\right)\\
    & = 3^{s-k_f-k_g}+(-1)^{s}\varepsilon_f\varepsilon_g{\sqrt{3^*}}^{s-k_f-k_g}\sum_{t\in \F_3^*}\zeta_3^{\lambda t^2}\\
    & = 3^{s-k_f-k_g}+(-1)^{s}\varepsilon_f\varepsilon_g{\sqrt{3^*}}^{s-k_f-k_g}\left(\eta_0(\lambda)\sqrt{3^*}-1\right).   
  \end{split}\]
 It follows from the definitions of $P_{f^*g^*}(1)$ and $P_{f^*g^*}(2)$ that 
 $$P=3P_{f^*g^*}(0)+\sqrt{3^*}P_{f^*g^*}(1)-\sqrt{3^*}P_{f^*g^*}(2).$$
  Combining the fact that $P_{f^*g^*}(0)+P_{f^*g^*}(1)+P_{f^*g^*}(2)=3^{s-k_f-k_g}$, 
  the desired results follow immediately straightforward from Lemma \ref{P_{f*,g*+c=0}}. 
\end{proof}

\begin{lemma}\label{Pg*}
  Let $g\in \mathcal{WRP}$ with $\widetilde{\mathcal{R}}_g(\beta)=\varepsilon_g{\sqrt{3^*}}^{m+k_g}\zeta_3^{g^*(\beta)}$ for every $\beta \in \widetilde{\mathcal{S}\mathcal{R}}_g$,
  where $\varepsilon_g\in \{ -1, 1\}$ and $0\leq k_g\leq m$.
  %Assume that $\lambda,\mu \in \F_3$.
  Let $\lambda\in \F_3^*$, $\mu \in \F_{3}$ and  
  \begin{align*}
      P_{g^*}(0)=\#\{\alpha \in \F_{3}^*, \beta \in \widetilde{\mathcal{S}\mathcal{R}}_g:g^*(\beta)+\lambda-\frac{\mu }{\alpha}=0\}.
  \end{align*} 
Then the following statements hold. 
\begin{enumerate}
  \item [\rm (1)] If $m+k_g$ is even, then 
  $$ P_{g^*}(0)=\left\{
        \begin{aligned}
        %&2\cdot 3^{m-k_g-1}+4\cdot\eta_0^{m+1}(-1)\varepsilon_g{(-3)}^{\frac{m-k_g-2}{2}}   && \lambda=0 \ {\rm and } \ \mu =0,\\    
        &2\cdot{3}^{m-k_g-1}-2\cdot(-1)^{m+1}\varepsilon_g{(-3)}^{\frac{m-k_g-2}{2}}   && \mu =0,\\ 
        %&& \makecell[l]{\lambda\neq 0 \ {\rm and }\ \mu=0  \\{\rm \ or \ }\lambda= 0 \ {\rm and } \ \mu\neq 0,}\\
        &2\cdot{3}^{m-k_g-1}+(-1)^{m+1}\varepsilon_g{(-3)}^{\frac{m-k_g-2}{2}}         && \mu\neq0.\\ 
        %&& \lambda\neq 0 \ {\rm and }\ \mu\neq 0,\\  
        \end{aligned}
        \right.
        $$
  \item [\rm (2)] If $m+k_g$ is odd, then 
  $$ P_{g^*}(0)=\left\{
    \begin{aligned}
      %&2\cdot3^{m-k_g-1}   && \lambda=0,\\    
      &2\cdot{3}^{m-k_g-1}+2\cdot(-1)^{m+1}\varepsilon_g{(-3)}^{\frac{m-k_g-1}{2}}\eta_0(\lambda) && \mu=0,\\  
      %&& \lambda\neq 0 \ {\rm and }\ \mu=0,\\
      &2\cdot{3}^{m-k_g-1}+(-1)^{m+1}\varepsilon_g{(-3)}^{\frac{m-k_g-1}{2}}\eta_0(2\lambda)    && \mu\neq 0.\\
      %&& \lambda\neq 0 \ {\rm and }\ \mu\neq 0,\\  
  \end{aligned} 
\right.
$$ 
\end{enumerate}
\end{lemma}
\begin{proof}
 By the definition of $P_{g^*}(0)$, we have 
    \[\begin{split}  
      P_{g^*}(0) &=\frac{1}{3}\sum_{t\in \F_3}\sum_{\alpha\in \F_{3}^*}\sum_{\beta\in \widetilde{\mathcal{S}\mathcal{R}}_g}
      \zeta_3^{t\left(g^*(\frac{\beta}{\alpha})+\lambda-\frac{\mu }{\alpha }\right)}\\    
           & =2\cdot 3^{m-k_g-1}+\frac{1}{3}\sum_{t\in \F_3^*}\sigma_{t}\left(\sum_{\alpha \in \F_{3}^*}\sum_{\beta\in \widetilde{\mathcal{S}\mathcal{R}}_g}\zeta_3^{g^*(\frac{\beta}{\alpha})+\lambda-\frac{\mu }{\alpha }} \right)\\     
           & =2\cdot 3^{m-k_g-1}+\frac{1}{3}\sum_{t\in \F_3^*}\sigma_{t}\left(\sum_{\beta\in \widetilde{\mathcal{S}\mathcal{R}}_g}{\zeta_3^{g^*(\beta)+\lambda-\mu}}+\sum_{\beta\in \widetilde{\mathcal{S}\mathcal{R}}_g}{\zeta_3^{g^*(\beta)+\lambda+\mu}}\right),\\
          \end{split}\]
  where the last equation holds since $\alpha \in \F_3^*$. 
 
  For $\mu = 0 $, we further have 
  \[\begin{split}
   P_{g^*}(0)&=2\cdot 3^{m-k_g-1}+\frac{2}{3}\sum_{t\in \F_3^*}\sigma_{t}\left(\sum_{\beta\in \widetilde{\mathcal{S}\mathcal{R}}_g}{\zeta_3^{g^*(\beta)+\lambda}}\right)\\
   &=2\cdot 3^{m-k_g-1}+2\cdot(-1)^{m+1}\varepsilon_g{(-3)}^{\frac{m-k_g-2}{2}}\sum_{t\in \F_3^*}\eta_0^{m-k_g}(t)\zeta_3^{\lambda t}\\ 
  &=
  \left\{
    \begin{aligned}
    &2\cdot3^{m-k_g-1}-2\cdot(-1)^{m+1}\varepsilon_g{(-3)}^{\frac{m-k_g-2}{2}}   && m+k_g {\rm \ is\ even},\\    
    &2\cdot{3}^{m-k_g-1}+2\cdot(-1)^{m+1}\varepsilon_g{(-3)}^{\frac{m-k_g-1}{2}}\eta_0(\lambda) &&  m+k_g {\rm \ is\ odd}.\\
    \end{aligned}
    \right.
  \end{split}\]
  For $\mu \neq 0$, it can also be checked that 
  \[\begin{split}
    P_{g^*}(0)&=2\cdot3^{m-k_g-1}+\frac{1}{3}\sum_{t\in \F_3^*}\sigma_{t}\left(\sum_{\beta\in \widetilde{\mathcal{S}\mathcal{R}}_g}{\zeta_3^{g^*(\beta)}+\sum_{\beta\in \widetilde{\mathcal{S}\mathcal{R}}_g}{\zeta_3^{g^*(\beta)+2\lambda}}}\right)\\
    &=2\cdot 3^{m-k_g-1}+(-1)^{m+1}\varepsilon_g{(-3)}^{\frac{m-k_g-2}{2}}\sum_{t\in \F_3^*}\eta_0^{m-k_g}(t)\left(1+\zeta_3^{2\lambda t}\right)\\ 
    &=
    \left\{
      \begin{aligned}
    &2\cdot3^{m-k_g-1}+(-1)^{m+1}\varepsilon_g{(-3)}^{\frac{m-k_g-2}{2}} && m+k_g \ {\rm  is\ even},\\
    &2\cdot 3^{m-k_g-1}+(-1)^{m+1}\varepsilon_g{(-3)}^{\frac{m-k_g-1}{2}}\eta_0(2\lambda) &&m+k_g {\rm \ is\ odd}.\\
      \end{aligned}
    \right.   
  \end{split}\]
  We have finished the proof. 
  \end{proof}

\begin{lemma}\label{S_1}
  Let $f,g\in \mathcal{WRP}$ with $\widetilde{\mathcal{R}}_f(\alpha )=\varepsilon_f{\sqrt{3^*}}^{n+k_f}\zeta_3^{f^*(\alpha)}$ for every $\alpha \in \widetilde{\mathcal{S}\mathcal{R}}_f$
  and $\widetilde{\mathcal{R}}_g(\beta)=\varepsilon_g{\sqrt{3^*}}^{m+k_g}\zeta_3^{g^*(\beta)}$ for every $\beta \in \widetilde{\mathcal{S}\mathcal{R}}_g$,
   where $\varepsilon_f,\varepsilon_g\in \{ -1, 1\}$, $0\leq k_f\leq n$ and $0\leq k_g\leq m$.
  Let $s=n+m$, $\lambda \in \F_3^*$ and 
  \[\begin{split}
    S_1= \sum_{t\in \F_3^*}\sum_{(x,y)\in \F_{3^n}\times \F_{3^m}}\zeta_3^{t(f(x)+g(y)+\lambda)}.
  \end{split}\]
  Then we have 
      $$ S_1=\left\{
        \begin{aligned}   
        &-\varepsilon_f\varepsilon_g{(-3)}^{\frac{s+k_f+k_g}{2}}    && s+k_f+k_g  {\rm \ is \ even},\\   
        &\varepsilon_f\varepsilon_g{(-3)}^{\frac{s+k_f+k_g+1}{2}} \eta_0(\lambda)   && s+k_f+k_g {\rm \ is \ odd}.\\
      \end{aligned}
        \right.
      $$
            \end{lemma}              
\begin{proof}
  It is clear that 
  \[\begin{split}
    S_1=& \sum_{t\in \F_3^*}\sum_{(x,y)\in \F_{3^n}\times \F_{3^m}}\zeta_3^{t(f(x)+g(y)+\lambda)}\\
     =&\sum_{t\in \F_3^*}\zeta_3^{\lambda t}\sigma_t\left(\sum_{x\in \F_{3^n}}\zeta_3^{f(x)}\sum_{y\in \F_{3^m}}\zeta_3^{g(x)}\right)\\ 
     =&\sum_{t\in \F_3^*}\zeta_3^{\lambda t}\varepsilon_f\varepsilon_g{\sqrt{3^*}}^{s+k_f+k_g}\eta_0^{s+k_f+k_g}(t).\\
    \end{split}\]  
The desired result then follows from Lemma \ref{characters}. 
\end{proof}

\begin{lemma}\label{S_2}
  
  Let $s=n+m$ and $\mu \in \F_3$, assume that 
  \[\begin{split}
    S_2= \sum_{u\in \F_3^*}\sum_{(x,y)\in \F_{3^n}\times \F_{3^m}}\zeta_3^{u(Tr_{3}^{s}(\alpha x+\beta y)+\mu)}.
  \end{split}\]
  For any  $(\alpha,\beta) \in \F_{3^n} \times \F_{3^m}$,  we have  
    $$ S_2=\left\{
          \begin{aligned}
          &0  && (\alpha, \beta)\neq (0, 0), \\    
          &2 \cdot 3^{s}   && (\alpha, \beta)=(0, 0)~{\rm and }~ \mu=0, \\
          &-3^{s}   && (\alpha, \beta)=(0, 0)~{\rm and }~ \mu\neq 0.\\ 
          \end{aligned}
          \right.
          $$
  \end{lemma}
  \begin{proof}
    It can be calculated that  
    \[\begin{split}
      S_2=& \sum_{u\in \F_3^*}\sum_{(x,y)\in \F_{3^n}\times \F_{3^m}}\zeta_3^{u(Tr_{3}^{s}(\alpha x+\beta y)+\mu)} \\ 
         =&\sum_{u\in \F_3^*}\zeta_3^{\mu u}\sigma_u\left(\sum_{x\in \F_{3^n}}\zeta_3^{Tr_{3}^{n}(\alpha x)}\sum_{y\in \F_{3^m}}\zeta_3^{Tr_{3}^{m}(\beta y)}\right)\\ 
       =&\left\{
        \begin{aligned}
        &0  && (\alpha, \beta)\neq (0, 0), \\    
        &2 \cdot 3^{s}   && (\alpha, \beta)=(0, 0)~{\rm and }~ \mu=0, \\
        &-3^{s}   && (\alpha, \beta)=(0, 0)~{\rm and }~ \mu\neq 0.\\ 
        \end{aligned}
        \right.
      \end{split}\] 
      This completes the proof. 
  \end{proof}

  \begin{lemma}\label{S_3}
   Let notations be the  same as Lemma \ref{S_1}. 
    Let $l_f$ and $l_g$ be two even positive integers satisfying 
  $f^*(a\alpha)=a^{l_f}f^*(\alpha)$ and  $g^*(b\beta)=b^{l_g}g^*(\beta)$,
  where $a,b\in \F_3^*$, $\alpha \in \widetilde{\mathcal{S}\mathcal{R}}_f$ and $\beta \in \widetilde{\mathcal{S}\mathcal{R}}_g$.
 Let $s=n+m$, $\lambda\in \F_3^*$, $\mu \in \F_3$ and 
  \[\begin{split}
    S_3= \sum_{t\in \F_3^*}\sum_{u \in \F_3^*}\sum_{(x,y)\in \F_{3^n}\times \F_{3^m}}\zeta_p^{t(f(x)+g(y)+\lambda)+u(Tr_{3}^{s}(\alpha x+\beta y)+\mu)}.
  \end{split}\]
 If $(\alpha,\beta) \notin \widetilde{\mathcal{S}\mathcal{R}}_f\times\widetilde{\mathcal{S}\mathcal{R}}_g$, then $S_3=0$.  
 If $(\alpha,\beta) \in \widetilde{\mathcal{S}\mathcal{R}}_f\times\widetilde{\mathcal{S}\mathcal{R}}_g$, then the following statements hold. 
    \begin{enumerate}
    \item [\rm (1)] If $s+k_f+k_g$ is even, then 
      $$ S_3=\left\{
              \begin{aligned}
              &4\varepsilon_f\varepsilon_g{(-3)}^{\frac{s+k_f+k_g}{2}}   &&f^*(\alpha)+g^*(\beta)+\lambda=0~{\rm and }~\mu=0, \\    
              &-2\varepsilon_f\varepsilon_g{(-3)}^{\frac{s+k_f+k_g}{2}}   &&\makecell[l]{f^*(\alpha)+g^*(\beta)+\lambda=0~{\rm and }~\mu \neq0, \\{\rm \ or \ }f^*(\alpha)+g^*(\beta)+\lambda\neq0~{\rm and }~\mu=0,}\\
              &\varepsilon_f\varepsilon_g{(-3)}^{\frac{s+k_f+k_g}{2}}   &&f^*(\alpha)+g^*(\beta)+\lambda\neq0~{\rm and }~\mu \neq0.
              \end{aligned}
              \right.
              $$
    
    \item [\rm (2)] If $s+k_f+k_g$ is odd, then 
            $$ S_3=\left\{
                    \begin{aligned}
                    &0   &&f^*(\alpha)+g^*(\beta)+\lambda=0, \\    
                    &2\varepsilon_f\varepsilon_g{(-3)}^{\frac{s+k_f+k_g+1}{2}}    &&f^*(\alpha)+g^*(\beta)+\lambda=1~{\rm and }~\mu =0,\\
                    &-2\varepsilon_f\varepsilon_g{(-3)}^{\frac{s+k_f+k_g+1}{2}}    &&f^*(\alpha)+g^*(\beta)+\lambda=2~{\rm and }~\mu =0,\\
                    &-\varepsilon_f\varepsilon_g{(-3)}^{\frac{s+k_f+k_g+1}{2}}    &&f^*(\alpha)+g^*(\beta)+\lambda=1~{\rm and }~\mu \neq0,\\
                    &\varepsilon_f\varepsilon_g{(-3)}^{\frac{s+k_f+k_g+1}{2}}    &&f^*(\alpha)+g^*(\beta)+\lambda=2~{\rm and }~\mu \neq0.   
                  \end{aligned}
                    \right.
                    $$            
    \end{enumerate}
    \end{lemma}
    \begin{proof}
      It is clear that 
    \[\begin{split}
      S_3=& \sum_{t\in \F_3^*}\sum_{u \in \F_3^*}\sum_{(x,y)\in \F_{3^n}\times \F_{3^m}}\zeta_3^{t(f(x)+g(y)+\lambda)+u(Tr_{3}^{s}(\alpha x+\beta y)+\mu )}\\
       =&\sum_{u\in \F_3^*}\zeta_3^{\mu u}\sum_{t\in \F_3^*}\zeta_3^{\lambda t}\sigma_t\left(\sum_{x\in \F_{3^n}}\zeta_3^{f(x)-Tr_{3}^{n}(-\frac{u}{t}\alpha x)}
       \sum_{y\in \F_{3^m}}\zeta_3^{g(x)-Tr_{3}^{m}(-\frac{u}{t}\beta x)}\right)\\ 
       =&\sum_{u\in \F_3^*}\zeta_3^{\mu u}\sum_{t\in \F_3^*}\zeta_3^{\lambda t}\sigma_t\left(\widetilde{\mathcal{R}}_f(-\frac{u}{t}\alpha)
       \widetilde{\mathcal{R}}_f(-\frac{u}{t}\alpha)\right).\\
      \end{split}\] 
    
    For $t,u\in \F_3^*$, if $(\alpha,\beta) \notin \widetilde{\mathcal{S}\mathcal{R}}_f\times\widetilde{\mathcal{S}\mathcal{R}}_g$,
    then $(-\frac{u \alpha}{t},-\frac{u \beta}{t}) \notin \widetilde{\mathcal{S}\mathcal{R}}_f\times\widetilde{\mathcal{S}\mathcal{R}}_g$
    and hence, $S_3=0$; 
    and if $(\alpha,\beta) \in \widetilde{\mathcal{S}\mathcal{R}}_f\times\widetilde{\mathcal{S}\mathcal{R}}_g$,
    then $(-\frac{u \alpha}{t},-\frac{u \beta}{t}) \in \widetilde{\mathcal{S}\mathcal{R}}_f\times\widetilde{\mathcal{S}\mathcal{R}}_g$  
    and it follows from the facts that $2\mid l_f$ and $2\mid l_g$ that 
    $(-\frac{u}{t})^{l_f}f^*(\alpha)=f^*(\alpha)$ and $(-\frac{u}{t})^{l_g}g^*(\beta)=g^*(\beta)$.
    It then implies that 
    \[\begin{split}
      S_3=&\sum_{u\in \F_3^*}\zeta_3^{-\mu u}\sum_{t\in \F_3^*}\varepsilon_f\varepsilon_g{\sqrt{3^*}}^{s+k_f+k_g}\eta_0^{s+k_f+k_g}(t)
      \zeta_3^{\lambda t}\zeta_3^{t\left((-\frac{u}{t})^{l_f}f^*(\alpha)+(-\frac{u}{t})^{l_g}g^*(\beta)\right)}\\ 
      =&\sum_{u\in \F_3^*}\zeta_3^{-\mu u}\sum_{t\in \F_3^*}\varepsilon_f\varepsilon_g{\sqrt{3^*}}^{s+k_f+k_g}\eta_0^{s+k_f+k_g}(t)
       \zeta_3^{t\left(f^*(\alpha)+g^*(\beta)+\lambda\right)}\\
    =&\left\{
        \begin{aligned}
        &4\varepsilon_f\varepsilon_g{(-3)}^{\frac{s+k_f+k_g}{2}}   &&f^*(\alpha)+g^*(\beta)+\lambda=0~{\rm and }~\mu=0, \\    
        &-2\varepsilon_f\varepsilon_g{(-3)}^{\frac{s+k_f+k_g}{2}}   &&\makecell[l]{f^*(\alpha)+g^*(\beta)+\lambda=0~{\rm and }~\mu \neq0, \\{\rm \ or \ }f^*(\alpha)+g^*(\beta)+\lambda\neq0~{\rm and }~\mu=0,}\\
        &\varepsilon_f\varepsilon_g{(-3)}^{\frac{s+k_f+k_g}{2}}   &&f^*(\alpha)+g^*(\beta)+\lambda\neq0~{\rm and }~\mu \neq0,
        \end{aligned}
        \right.
   \end{split}\]
   when $s+k_f+k_g$ is even. This completes the desired result (1). 
   Taking the same argument as that of the proof of the result (1), the expected result (2) also holds. 
   We have finished the whole proof. 
   %Otherwise, 
    %$$ S_3=\left\{
    %  \begin{aligned}
    %  &0   &&f^*(\alpha)+g^*(\beta)+\lambda=0, \\    
    %  &2\varepsilon_f\varepsilon_g{(-3)}^{\frac{s+k_f+k_g+1}{2}}  &&\mu =0,f^*(\alpha)+g^*(\beta)+\lambda=1,\\
    %  &-2\varepsilon_f\varepsilon_g{(-3)}^{\frac{s+k_f+k_g+1}{2}}  &&\mu =0,f^*(\alpha)+g^*(\beta)+\lambda=2,\\
    %  &-\varepsilon_f\varepsilon_g{(-3)}^{\frac{s+k_f+k_g+1}{2}}   &&\mu \neq0,f^*(\alpha)+g^*(\beta)+\lambda=1,\\
    %  &\varepsilon_f\varepsilon_g{(-3)}^{\frac{s+k_f+k_g+1}{2}}  &&\mu \neq0,f^*(\alpha)+g^*(\beta)+\lambda=2.  
    %\end{aligned}
    %  \right.
    %  $$       
    \end{proof}

      \begin{lemma}\label{S_4}
       Let notations be the  same as Lemma \ref{Pg*}. 
     Let $s=n+m$, $\lambda\in \F_3^*$, $\mu \in \F_3$ and 
      \[\begin{split}
        S_4= \sum_{t\in \F_3^*}\sum_{u \in \F_3^*}\sum_{(x,y)\in \F_{3^n}\times \F_{3^m}}\zeta_3^{t(Tr_{3}^{n}(x)+g(y)+\lambda)+u(Tr_{3}^{s}(\alpha x+\beta y)+\mu)}.
      \end{split}\]
      For any $(\alpha,\beta) \in \F_{3^n}\times\F_{3^m} \backslash\{(0,0)\}$, the following statements hold. 
      \begin{enumerate}
        \item [\rm (1)] If $m+k_g$ is even, then 
          $$ S_4=\left\{
                  \begin{aligned}
                  &2 \cdot 3^n\varepsilon_g{(-3)}^{\frac{m+k_g}{2}}   &&\alpha \in \F_3^*,~ \beta \in \widetilde{\mathcal{S}\mathcal{R}}_g \ {\rm and } \ g^*(\frac{\beta}{\alpha})+\lambda-\frac{\mu}{\alpha}=0, \\    
                  &-3^n\varepsilon_g{(-3)}^{\frac{m+k_g}{2}}       &&\alpha \in \F_3^*,~ \beta \in \widetilde{\mathcal{S}\mathcal{R}}_g \ {\rm and } \ g^*(\frac{\beta}{\alpha})+\lambda-\frac{\mu}{\alpha}\neq 0,\\
                  &0   &&{\rm otherwise}.
                  \end{aligned}
                  \right.
                  $$
        \item [\rm (2)] If $m+k_g$ is odd, then 
        $$ S_4=\left\{
          \begin{aligned}
          &3^n\varepsilon_g{(-3)}^{\frac{m+k_g+1}{2}}  &&\alpha \in \F_3^*,~ \beta \in \widetilde{\mathcal{S}\mathcal{R}}_g \ {\rm and } \ g^*(\frac{\beta}{\alpha})+\lambda-\frac{\mu}{\alpha} =1, \\    
          &-3^n\varepsilon_g{(-3)}^{\frac{m+k_g+1}{2}} &&\alpha \in \F_3^*,~ \beta \in \widetilde{\mathcal{S}\mathcal{R}}_g \ {\rm and } \ g^*(\frac{\beta}{\alpha})+\lambda-\frac{\mu}{\alpha} =2, \\    
          &0                                                                                                   &&{\rm otherwise}.
          \end{aligned}
          \right.
          $$           
        \end{enumerate}
        \end{lemma}
        \begin{proof}
          It is clear that 
        \[\begin{split}
          S_4
          %=&\sum_{t\in \F_3^*}\sum_{u \in \F_3^*}\sum_{(x,y)\in \F_{3^n}\times \F_{3^m}}\zeta_3^{t(f(x)+g(y)+\lambda)+u(Tr(\alpha x+\beta y)+\mu )}\\
           =&\sum_{u\in \F_3^*}\zeta_3^{\mu u}\sum_{t\in \F_3^*}\zeta_3^{\lambda t}\sigma_t\left(\sum_{x\in \F_{3^n}}\zeta_3^{Tr_{3}^{n}(x)+\frac{u}{t}Tr_{3}^{n}(\alpha x)}
           \sum_{y\in \F_{3^m}}\zeta_3^{g(y)+\frac{u}{t}Tr_{3}^{m}(\beta y)}\right)\\ 
           =&\sum_{t\in \F_3^*}\zeta_3^{\lambda t}\sigma_t\left(\sum_{u\in \F_3^*}\zeta_3^{\mu \frac{u}{t}}\sum_{x\in \F_{3^n}}\zeta_3^{Tr_{3}^{n}(x)+\frac{u}{t}Tr_{3}^{n}(\alpha x)}
           \sum_{y\in \F_{3^m}}\zeta_3^{g(y)+\frac{u}{t}Tr_{3}^{m}(\beta y)}\right)\\
           %=&p^n\sum_{t \in \F_p^*}\zeta_p^{\lambda t}\sigma_t\left(\sum_{u\in \F_3^*}\zeta_3^{-\frac{\mu}{\alpha}}\sum_{y\in \F_{3^m}}\zeta_3^{g(y)-Tr(\frac{\beta}{\alpha}y)}\right)\\
           %=&p^n\varepsilon_g{\sqrt{p^*}}^{m+k_g}\sum_{t \in \F_p^*}\eta_0^{m+k_g}(t)\zeta_3^{t\left(g^*(\frac{\beta}{\alpha})+\lambda-\frac{\mu}{\alpha}\right)}\\
          =&\left\{
            \begin{aligned}
            &3^n\sum_{t \in \F_3^*}\zeta_3^{\lambda t}\sigma_t\left(\sum_{u\in \F_3^*}\zeta_3^{-\frac{\mu}{\alpha}}\sum_{y\in \F_{3^m}}\zeta_3^{g(y)-Tr_{3}^{m}(\frac{\beta}{\alpha}y)}\right)  &&\alpha \in \F_3^*, \\    
            &0                                                                                                   &&\alpha \notin \F_3^*,\\
            \end{aligned}
            \right.\\
            =&\left\{
              \begin{aligned}
                &3^n\varepsilon_g{(-3)}^{\frac{m+k_g}{2}}\sum_{t \in \F_p^*}\eta_0^{m+k_g}(t)\zeta_3^{t\left(g^*(\frac{\beta}{\alpha})+\lambda-\frac{\mu}{\alpha}\right)}  &&\alpha \in \F_3^*, \\    
                &0                                                                                                   &&\alpha \notin \F_3^*.\\
              \end{aligned}
              \right.  
          \end{split}\] 
        When $m+k_g$ is even, then we have 
        \[\begin{split}
          S_4=\left\{
            \begin{aligned}
            &2\cdot 3^n\varepsilon_g{(-3)}^{\frac{m+k_g}{2}}  &&\alpha \in \F_3^* \ {\rm and}\  g^*(\frac{\beta}{\alpha})+\lambda-\frac{\mu}{\alpha}=0,   \\    
            &-3^n\varepsilon_g{(-3)}^{\frac{m+k_g}{2}}  &&\alpha \in \F_3^* \ {\rm and}\  g^*(\frac{\beta}{\alpha})+\lambda-\frac{\mu}{\alpha}\neq 0,   \\ 
            &0                                                                                                   &&{\rm otherwise}.\\
            \end{aligned}
            \right.
        \end{split}\]
        When $m+k_g$ is odd, then we have 
        \[\begin{split}
          S_4=\left\{
            \begin{aligned}
            &3^n\varepsilon_g{(-3)}^{\frac{m+k_g+1}{2}}  &&\alpha \in \F_3^* \ {\rm and}\  g^*(\frac{\beta}{\alpha})+\lambda-\frac{\mu}{\alpha}=1,   \\     
            &-3^n\varepsilon_g{(-3)}^{\frac{m+k_g+1}{2}}   &&\alpha \in \F_3^* \ {\rm and}\  g^*(\frac{\beta}{\alpha})+\lambda-\frac{\mu}{\alpha}=2,   \\     
            &0                                                                                                   &&{\rm otherwise}.\\
            \end{aligned}
            \right.
        \end{split}\]
        This completes the proof. 
      \end{proof}

        \begin{lemma}\label{S_5}
          Let notations be the  same as Lemma \ref{Pg*}. 
       Let $s=n+m$, $\lambda\in \F_3^*$, $\mu \in \F_3$ and 
        \[\begin{split}
          S_5= \sum_{t\in \F_3^*}\sum_{u \in \F_3^*}\sum_{(x,y)\in \F_{3^n}\times \F_{3^m}}\zeta_3^{t^2(Tr_{3}^{n}(x)+g(y)+\lambda)+u(Tr_{3}^{s}(\alpha x+\beta y)+\mu)}.
        \end{split}\]
        For any $(\alpha,\beta) \in \F_{3^n}\times\F_{3^m} \backslash\{(0,0)\}$, 
        the following statements hold.   
              $$ S_5=\left\{
                    \begin{aligned}
                    &2\cdot 3^n\varepsilon_g{(-3)}^{\frac{m+k_g}{2}}   &&\alpha \in \F_3^*, \beta \in \widetilde{\mathcal{S}\mathcal{R}}_g \ {\rm and } \ g^*(\frac{\beta}{\alpha})+\lambda-\frac{\mu}{\alpha}=0, \\    
                    &({\sqrt{-3}}-1)3^n\varepsilon_g{(-3)}^{\frac{m+k_g}{2}}       &&\alpha \in \F_3^*, \beta \in \widetilde{\mathcal{S}\mathcal{R}}_g \ {\rm and } \ g^*(\frac{\beta}{\alpha})+\lambda-\frac{\mu}{\alpha}=1,\\
                    &({\sqrt{-3}}+1)3^n\varepsilon_g{(-3)}^{\frac{m+k_g}{2}}       &&\alpha \in \F_3^*, \beta \in \widetilde{\mathcal{S}\mathcal{R}}_g \ {\rm and } \ g^*(\frac{\beta}{\alpha})+\lambda-\frac{\mu}{\alpha}=2,\\
                    &0   &&{\rm otherwise}.
                    \end{aligned}
                    \right.
                    $$
               \end{lemma}
         
               \begin{proof}
            The proof is similar to that of Lemma \ref{S_4} and we immediately have 
            \[\begin{split}
              S_5=&\left\{
                  \begin{aligned}
                    &3^n\varepsilon_g{(-3)}^{\frac{m+k_g}{2}}\sum_{t \in \F_3^*}\eta_0^{m+k_g}(t^2)\zeta_3^{t^2\left(g^*(\frac{\beta}{\alpha})+\lambda-\frac{\mu}{\alpha}\right)}  &&\alpha \in \F_3^*, \\    
                    &0                                                                                                   &&\alpha \notin \F_3^*,\\
                  \end{aligned}
                  \right.\\
                  =&\left\{
                \begin{aligned}
                &2\cdot3^n\varepsilon_g{(-3)}^{\frac{m+k_g}{2}}  &&\alpha \in \F_3^* \ {\rm and}\  g^*(\frac{\beta}{\alpha})+\lambda-\frac{\mu}{\alpha}=0,   \\    
                &3^n\varepsilon_g{(-3)}^{\frac{m+k_g}{2}}\left(\eta_0\left(g^*(\frac{\beta}{\alpha})+\lambda-\frac{\mu}{\alpha}\right){\sqrt{-3}}-1\right)  &&\alpha \in \F_3^* \ {\rm and}\  g^*(\frac{\beta}{\alpha})+\lambda-\frac{\mu}{\alpha}\neq 0,   \\ 
                &0                                      &&{\rm otherwise}.\\
              \end{aligned}
              \right. 
              \end{split}\]  
              The desired result immediately then follows from the definition of $\eta_0$. 
            \end{proof}

      \begin{lemma}\label{Ng(0)}
        Let $g\in \mathcal{WRP}$ with $\widetilde{\mathcal{R}}_g(\beta)=\varepsilon_g{\sqrt{3^*}}^{m+k_g}\zeta_3^{g^*(\beta)}$ for every $\beta \in \widetilde{\mathcal{S}\mathcal{R}}_g$,
      where $\varepsilon_g\in \{ -1, 1\}$ and $0\leq k_g\leq m$.
        Let $s=n+m$, $\lambda\in \F_3^*$, $\mu \in \F_{3}$ and 
        \begin{align*}
            N_{g}(0)=\#\{(x,y)\in \F_{3^n}\times \F_{3^m}: Tr_{3}^{n}(x)+g(y)+\lambda=0,~ Tr_{3}^{s}(\alpha x +\beta y)+\mu=0\}.    
        \end{align*} 
        For any $(\alpha,\beta) \in \F_{3^n}\times\F_{3^m} \backslash\{(0,0)\}$, the following statements hold.  
      \begin{enumerate}
        \item [\rm (1)] If $m+k_g$ is even, then 
        \[\begin{split} 
        N_{g}(0)&=\left\{
              \begin{aligned}
              &3^{s-2}+2\cdot 3^{n-2}\varepsilon_g{(-3)}^{\frac{m+k_g}{2}}   && \alpha \in \F_3^*,~ \beta \in \widetilde{\mathcal{S}\mathcal{R}}_g \ {\rm and } \ g^*(\frac{\beta}{\alpha})+\lambda-\frac{\mu}{\alpha}=0,\\    
              &3^{s-2}-3^{n-2}\varepsilon_g{(-3)}^{\frac{m+k_g}{2}}   && \alpha \in \F_3^*,~ \beta \in \widetilde{\mathcal{S}\mathcal{R}}_g \ {\rm and } \ g^*(\frac{\beta}{\alpha})+\lambda-\frac{\mu}{\alpha}\neq 0,\\
              &3^{s-2}  && {\rm otherwise}.\\  
              \end{aligned}
              \right.
            \end{split}\]
        \item [\rm (2)] If $m+k_g$ is odd, then 
        \[\begin{split}
        N_{g}(0)&=\left\{
              \begin{aligned}
              &3^{s-2}+3^{n-2}\varepsilon_g{(-3)}^{\frac{m+k_g+1}{2}}   && \alpha \in \F_3^*,~ \beta \in \widetilde{\mathcal{S}\mathcal{R}}_g \ {\rm and } \ g^*(\frac{\beta}{\alpha})+\lambda-\frac{\mu}{\alpha}=1,\\    
              &3^{s-2}-3^{n-2}\varepsilon_g{(-3)}^{\frac{m+k_g+1}{2}}   && \alpha \in \F_3^*,~ \beta \in \widetilde{\mathcal{S}\mathcal{R}}_g \ {\rm and } \ g^*(\frac{\beta}{\alpha})+\lambda-\frac{\mu}{\alpha}=2,\\
              &3^{s-2}  && {\rm otherwise}.\\  
              \end{aligned}
              \right.
            \end{split}\]
    
      \end{enumerate}
      \end{lemma}  
    \begin{proof}
      It is clear that 
      \[\begin{split}
        N_{g}(0)=&\frac{1}{9}\sum_{t \in \F_3}\sum_{u \in \F_3}\zeta_3^{t(Tr_{3}^{n}(x)+g(y)+\lambda)+u(Tr_{3}^{s}(\alpha x+\beta y)+\mu )}\\
                              =&3^{s-2}+\frac{1}{9}\sum_{t \in \F_3^*}\sum_{(x, y) \in \F_{3^n}\times \F_{3^m}}\zeta_3^{t(Tr_{3}^{n}(x)+g(y)+\lambda)}
                              +\frac{1}{9}\sum_{u \in \F_3^*}\sum_{(x, y) \in \F_{3^n}\times \F_{3^m}}\zeta_3^{u(Tr_{3}^{s}(\alpha x+\beta y)+\mu)}\\
                              &+\frac{1}{9}\sum_{t \in \F_3^*}\sum_{u \in \F_3^*}\sum_{(x, y) \in \F_{3^n}\times \F_{3^m}}\zeta_3^{t(Tr_{3}^{n}(x)+g(y)+\lambda)+u(Tr_{3}^{s}(\alpha x+\beta y)+\mu)}\\
                              =&3^{s-2}+\frac{1}{9}S_2+\frac{1}{9}S_4.                            
    \end{split}\]
    Then the desired results follow from Lemmas \ref{S_2} and \ref{S_4}.
    \end{proof}

\section{Several new families of ternary self-orthogonal codes}\label{sec.SO}
%In this section, we let $p=3$ and $s=n+m$, where $n$ and $m$ are two positive integers.
In this section, we present several new explicit constructions of infinite families of ternary self-orthogonal codes 
by using weakly regular plateaued functions ($\mathcal{WRP}$).  
To this end, we need the following useful result.

\begin{lemma}[\cite{ternary self-orthogonal}]\label{lem.SO}
  Let $\mathcal{C}$ be any ternary linear code. 
  Then $\mathcal{C}$ is self-orthogonal if every codeword of $\mathcal{C}$ has weight which is divisible by $3$. 
\end{lemma}

\subsection{Three new families of ternary self-orthogonal codes from $f,g\in \mathcal{WRP}$} 

In this subsection, we consider the augmented code $\overline{\mathcal{C}_{D_{fg}(0)}}$ defined in Equation (\ref{df: double AC_D}) with the defining set $D_{fg}(0)$ given by 
\begin{align}
D_{fg}(0)=\left\{(x,y)\in \F_{3^n}\times \F_{3^m}: f(x)+g(y)+\lambda=0\right\},
\end{align}
where $f,g\in \mathcal{WRP}$ and $\lambda \in \F_3^*$.

    Let $n_{fg}(0)$ denote the length of the augmented code $\overline{\mathcal{C}_{D_{fg}(0)}}$ and 
    ${\rm wt}(\overline{{\bf{c}}}_{fg}(0))$ denote the weight of nonzero codeword $\overline{{\bf{c}}}_{g}(0)$. 
    Let $s=n+m$, $\lambda \in \F_3^*$, $\mu  \in \F_3$ and   
    \begin{align}\label{df.N(lambda,mu)}
    N_{fg}(0)=\#\left\{(x,y)\in \F_{3^n}\times \F_{3^m}:f(x)+g(y)+\lambda=0,Tr_{3}^{s}(\alpha x+\beta y)+\mu=0\right\} 
  \end{align} 
  for any $(\alpha,\beta)\in \F_{3^n}\times \F_{3^m}$. 
  Then it can be verified that 
    \begin{align}\label{eq.n_fg(0)}
      n_{fg}(0)=\#\left\{(x,y)\in \F_{3^n}\times \F_{3^m}:f(x)+g(y)+\lambda=0\right\}
    \end{align}
    %Then by the definition of ${\rm wt}(\overline{{\bf{c}}}_{fg}(0))$, we have 
    and 
    \begin{align*}
    {\rm wt}(\overline{{\bf{c}}}_{fg}(0))=&n_{fg}(0)-N_{fg}(0)\\
    %=& \frac{1}{p}\sum_{t\in \F_p}\sum_{(x,y)\in \F_{p^n}\times \F_{p^m}}\zeta_p^{t(f(x)+g(y)+\lambda)}
    %-\frac{1}{p^2}\sum_{t\in \F_p}\sum_{u \in \F_p}\sum_{(x,y)\in \F_{p^n}\times \F_{p^m}}\zeta_p^{t(f(x)+g(y)+\lambda)+u(Tr(\alpha x+\beta y)+\mu)}\\
    =&\sum_{(x,y)\in \F_{3^n}\times \F_{3^m}}\left(\frac{1}{3}\sum_{t\in \F_3}\zeta_3^{t(f(x)+g(y)+\lambda)}
    -\frac{1}{9}\sum_{t\in \F_3}\sum_{u \in \F_3}\zeta_3^{t(f(x)+g(y)+\lambda)+u(Tr_{3}^{s}(\alpha x+\beta y)+\mu)}\right)\\
    %=& p^{s-1}+\frac{1}{p}\sum_{t\in \F_p^*}\sum_{(x,y)\in \F_{p^n}\times \F_{p^m}}\zeta_p^{t(f(x)+g(y)+\lambda)}
    %-p^{s-2}-\frac{1}{p^2}\sum_{t\in \F_p^*}\sum_{(x,y)\in \F_{p^n}\times \F_{p^m}}\zeta_p^{t(f(x)+g(y)+\lambda)}\\
    %-&\frac{1}{p^2}\sum_{u\in \F_p^*}\sum_{(x,y)\in \F_{p^n}\times \F_{p^m}}\zeta_p^{u(Tr(\alpha x+\beta y)+\mu )}-\frac{1}{p^2}\sum_{t\in \F_p}\sum_{u \in \F_p}\sum_{(x,y)\in \F_{p^n}\times \F_{p^m}}\zeta_p^{t(f(x)+g(y)+\lambda)+u(Tr(\alpha x+\beta y)+\mu)}\\
    =& 2\cdot 3^{s-2}+\frac{2}{9}S_1-\frac{1}{9}S_2-\frac{1}{9}S_3,
    \end{align*}
where 
\begin{align*}
  S_1=& \sum_{t\in \F_3^*}\sum_{(x,y)\in \F_{3^n}\times \F_{3^m}}\zeta_3^{t(f(x)+g(y)+\lambda)},\\
  S_2=& \sum_{u\in \F_3^*}\sum_{(x,y)\in \F_{3^n}\times \F_{3^m}}\zeta_3^{u(Tr_{3}^{s}(\alpha x+\beta y)+\mu)}~{\rm and} \\
  S_3=& \sum_{t\in \F_3^*}\sum_{u \in \F_3^*}\sum_{(x,y)\in \F_{3^n}\times \F_{3^m}}\zeta_3^{t(f(x)+g(y)+\lambda)+u(Tr_{3}^{s}(\alpha x+\beta y)+\mu)}.\\
\end{align*} 

\begin{theorem}\label{Th. double even ternary}
Let $s=n+m$, where $n$ and $m$ are two positive integers.
Let $f,g\in \mathcal{WRP}$ with $\widetilde{\mathcal{R}}_f(\alpha )=\varepsilon_f{\sqrt{3^*}}^{n+k_f}\zeta_3^{f^*(\alpha)}$ for every $\alpha \in \widetilde{\mathcal{S}\mathcal{R}}_f$
  and $\widetilde{\mathcal{R}}_g(\beta)=\varepsilon_g{\sqrt{3^*}}^{m+k_g}\zeta_3^{g^*(\beta)}$ for every $\beta \in \widetilde{\mathcal{S}\mathcal{R}}_g$,
   where $\varepsilon_f,\varepsilon_g\in \{ -1, 1\}$, $0\leq k_f\leq n$ and $0\leq k_g\leq m$.
   Assume that $s+k_f+k_g$ is even and $k_f+k_g<s-3$,
   then $\overline{\mathcal{C}_{D_{fg}(0)}}$ is a ternary five-weight self-orthogonal 
   $[3^{s-1}+\varepsilon_f\varepsilon_g(-3)^{\frac{s+k_f+k_g-2}{2}},s+1]$ code and its weight distribution is given in Table \ref{tab: double even }. 
\end{theorem}
\begin{proof}
  By the definition of $n_{fg}(0)$ in Equation (\ref{eq.n_fg(0)}) and Lemma \ref{S_1}, 
  we have $n_{fg}(0)=3^{s-1}+\varepsilon_f\varepsilon_g(-3)^{\frac{s+k_f+k_g-2}{2}}$ when $s+k_f+k_g$ is even.
  It clear the dimension of $\overline{\mathcal{C}_{D_{fg}(0)}}$ is $s+1$, 
  then this is a $[3^{s-1}+\varepsilon_f\varepsilon_g(-3)^{\frac{s+k_f+k_g-2}{2}},s+1]$ code.
  It follows from the definition of ${\rm wt}(\overline{{\bf{c}}}_{fg}(0))$ and Lemmas \ref{S_1}, \ref{S_2} and \ref{S_3} that 
  if   $(\alpha, \beta) \notin\widetilde{\mathcal{S}\mathcal{R}}_f\times\widetilde{\mathcal{S}\mathcal{R}}_g \backslash \{(0,0)\}$, 
  then ${\rm wt}(\overline{{\bf{c}}}_{fg}(0))=2\cdot3^{s-2}-2\varepsilon_f\varepsilon_g{(-3)}^{\frac{s+k_f+k_g-2}{2}}$; 
  and if $(\alpha, \beta) \in\widetilde{\mathcal{S}\mathcal{R}}_f\times\widetilde{\mathcal{S}\mathcal{R}}_g\bigcup\{(0,0)\}$, then 
  $$ {\rm wt}(\overline{{\bf{c}}}_{fg}(0))=\left\{
      \begin{aligned}
      &0   &&\mu=0~{\rm and}~(\alpha,\beta)=(0,0), \\    
      &2\cdot3^{s-2}+2\varepsilon_f\varepsilon_g{(-3)}^{\frac{s+k_f+k_g-2}{2}}   &&\mu=0~{\rm and}~f^*(\alpha)+g^*(\beta)+\lambda=0,\\
      &2\cdot3^{s-2}+\varepsilon_f\varepsilon_g{(-3)}^{\frac{s+k_f+k_g-2}{2}}   &&\makecell[l]{\mu=0,f^*(\alpha)+g^*(\beta)+\lambda\neq0,(\alpha,\beta) \neq 0,}\\
      &2\cdot 3^{s-2} &&\makecell[l]{\mu=0,f^*(\alpha)+g^*(\beta)+\lambda\neq0,(\alpha,\beta) \neq 0\\{\rm \ or \ }\mu \neq 0, f^*(\alpha)+g^*(\beta)+\lambda=0,(\alpha,\beta) \neq 0,}\\
      &3^{s-1}+\varepsilon_f\varepsilon_g{(-3)}^{\frac{s+k_f+k_g-2}{2}}   &&\mu \neq0~{\rm and}~(\alpha,\beta)=(0,0).\\  
      %&2\cdot3^{s-2}-\frac{2}{9}\varepsilon_f\varepsilon_g{\sqrt{3^*}}^{s+k_f+k_g} && (\alpha, \beta) \notin\widetilde{\mathcal{S}\mathcal{R}}_f\times\widetilde{\mathcal{S}\mathcal{R}}_g \backslash \{(0,0)\}.   
    \end{aligned}
      \right.
      $$ 
   The weight distribution of the code is straightly derived from Lemma \ref{P_{f*,g*+c=0}} and \ref{P_{f*+g*+c}}.

      Since $s+k_f+k_g$ is even and $k_f+k_g<s-3$, then $s-2>1$ and $\frac{s+k_f+k_g-2}{2}\geq 1$ is an integer. 
  It implies that $3\mid {\rm wt}(\overline{{\bf{c}}}_{fg}(0))$ for any ${\overline{{\bf{c}}}_{fg}(0)}\in \overline{\mathcal{C}_{D_{fg}(0)}}$. 
  Then it turns out from Lemma \ref{lem.SO} that $\overline{\mathcal{C}_{D_{fg}(0)}}$ is self-orthogonal and we complete the whole proof.  
\end{proof}

  \begin{table}[H]
    % table caption is above the table
    \caption{The weight distribution of $\overline{\mathcal{C}_{D_{fg}(0)}}$ when $s+k_f+k_g$ is even.} 
    \label{tab: double even }       % Give a unique label
    % For LaTeX tables use
    \centering
    
    \begin{tabular}{clllc}
            \hline
         &Weight $i$ & Multiplicity $A_i$ \\
            \hline
            \hline
             &$0$ & $1$ \\
             &$2\cdot3^{s-2}+2\varepsilon_f\varepsilon_g{(-3)}^{\frac{s+k_f+k_g-2}{2}}$ & $3^{s-k_f-k_g-1}-(-1)^{s+1}\varepsilon_f\varepsilon_g{(-3)}^{\frac{s-k_f-k_g-2}{2}}$ \\
             &$2\cdot3^{s-2}+\varepsilon_f\varepsilon_g{(-3)}^{\frac{s+k_f+k_g-2}{2}}$ & $4\cdot3^{s-k_f-k_g-1}+2(-1)^{s+1}\varepsilon_f\varepsilon_g{(-3)}^{\frac{s-k_f-k_g-2}{2}}-2$ \\
             &$2\cdot 3^{s-2}$ & $4 \cdot 3^{s-k_f-k_g-1}-(-1)^{s+1}\varepsilon_f\varepsilon_g{(-3)}^{\frac{s-k_f-k_g-2}{2}}-1$ \\ 
             &$3^{s-1}+\varepsilon_f\varepsilon_g{(-3)}^{\frac{s+k_f+k_g-2}{2}}$ & $2$\\
             &$2\cdot3^{s-2}-2\varepsilon_f\varepsilon_g{(-3)}^{\frac{s+k_f+k_g-2}{2}}$ &$3^{s+1}-3^{s-k_f-k_g+1}$\\
             \hline      
            \end{tabular}
      \end{table}
  
      \begin{example}\label{4+2+1,3+1+-1, even}
        Let $n=4$ and $m=3$, then $s=n+m=7$.  Let $f(x)=Tr_{3}^{4}(2x^{92})$ and  $g(x)=Tr_{3}^{3}(\omega x^{13}+\omega^7x^4+\omega^7x^3+\omega x^2)$, 
        where $\omega$ is a primitive element of $\F_{3^3}$ with $\omega^3+2\omega+1=0$. 
        On the one hand, it can be checked that $f(x)$ is a ternary weakly regular 2-plateaued function with $\widetilde{\mathcal{R}}_f(\alpha )\in \{0,-3^3\zeta_3^{f^*(\alpha)}\}$ 
        for all $\alpha \in \F_{3^3}$, where the sign of the Walsh transform of $f(x)$ is $\varepsilon_f=1$ and $f^*$ is an unbalanced function over $\F_{3^3}$ satisfying $f^*(0)=0$.
        On the other hand, it follows that $g(x)$ is a ternary weakly regular 1-plateaued function and the sign of the Walsh transform of $g(x)$ is $\varepsilon_g=-1$. 
        Therefore, it turns out from Theorem \ref{Th. double odd ternary} that for $\lambda\in \F_3^*$, $\overline{\mathcal{C}_{D_{fg}(0)}}$ 
        is a ternary four-weight self-orthogonal $[648,8,486]$ code and its weight enumerator is 
        $$1+24x^{324}+112x^{405}+104x^{486}+6320x^{648}.$$ 
        Verified by Magma \cite{Magma}, these results are true. 
      \end{example}

  \begin{theorem}\label{Th. double odd ternary}
    Let $s=n+m$, where $n$ and $m$ are two positive integers.
    Let $f,g\in \mathcal{WRP}$ with $\widetilde{\mathcal{R}}_f(\alpha )=\varepsilon_f{\sqrt{3^*}}^{n+k_f}\zeta_3^{f^*(\alpha)}$ 
    for every $\alpha \in \widetilde{\mathcal{S}\mathcal{R}}_f$ 
    and $\widetilde{\mathcal{R}}_g(\beta)=\varepsilon_g{\sqrt{3^*}}^{m+k_g}\zeta_3^{g^*(\beta)}$ 
    for every $\beta \in \widetilde{\mathcal{S}\mathcal{R}}_g$,    
  where $\varepsilon_f,\varepsilon_g\in \{ -1, 1\}$, $0\leq k_f\leq n$ and $0\leq k_g\leq m$.
  Assume that $s+k_f+k_g$ is odd and $k_f+k_g<s-3$,  
  then the following statements hold. 
  \begin{enumerate}
    \item [\rm (1)] If $\lambda=1$, then $\overline{\mathcal{C}_{D_{fg}(0)}}$ is a ternary six-weight self-orthogonal 
    $[3^{s-1}-\varepsilon_f\varepsilon_g{(-3)}^{\frac{s+k_f+k_g-1}{2}},s+1]$ code and its weight distribution is given in Table \ref{tab:double odd 1}. 
  
    \item [\rm (2)] If $\lambda=2$, then $\overline{\mathcal{C}_{D_{fg}(0)}}$ is a ternary six-weight self-orthogonal 
    $[3^{s-1}+\varepsilon_f\varepsilon_g{{(-3)}}^{\frac{s+k_f+k_g-1}{2}},s+1]$ code and its weight distribution is given in Table \ref{tab:double odd -1}. 
  \end{enumerate}
  \end{theorem}
  \begin{proof}
  %For any $(\alpha,\beta) \in \F_{p^n}\times \F_{p^m}$, $\lambda \in \F_p^*$ and $\mu \in \F_p^*$, 
  %it follows from the definition of ${\rm wt}(\overline{{\bf{c}}}_{fg}(0))$ and Lemmas \ref{S_1}, \ref{S_2} and \ref{S_3} that 

   (1) Since $\lambda=1$,  
  then by the definition of $n_{fg}(0)$ in Equation (\ref{eq.n_fg(0)}) and Lemma \ref{S_1}, 
  we have $n_{fg}(0)=3^{s-1}-\varepsilon_f\varepsilon_g(-3)^{\frac{s+k_f+k_g-1}{2}}$ when $s+k_f+k_g$ is odd.
  It clear the dimension of $\overline{\mathcal{C}_{D_{fg}(0)}}$ is $s+1$, 
  then this is a $[3^{s-1}-\varepsilon_f\varepsilon_g(-3)^{\frac{s+k_f+k_g-1}{2}},s+1]$ code.
   It follows from the definition of ${\rm wt}(\overline{{\bf{c}}}_{fg}(0))$ and Lemmas \ref{S_1}, \ref{S_2} and \ref{S_3} that  
   if $(\alpha, \beta) \notin\widetilde{\mathcal{S}\mathcal{R}}_f\times\widetilde{\mathcal{S}\mathcal{R}}_g \backslash \{(0,0)\}$, 
   then ${\rm wt}(\overline{{\bf{c}}}_{fg}(0))=2\cdot3^{s-2}+2\varepsilon_f\varepsilon_g{(-3)}^{\frac{s+k_f+k_g-3}{2}}$; 
   and if $(\alpha, \beta) \in\widetilde{\mathcal{S}\mathcal{R}}_f\times\widetilde{\mathcal{S}\mathcal{R}}_g\bigcup\{(0,0)\}$, then 
  $$ {\rm wt}(\overline{{\bf{c}}}_{fg}(0))=\left\{
        \begin{aligned}
        &0   &&\mu =0~{\rm and}~(\alpha,\beta)=(0,0), \\    
        &2\cdot3^{s-2}+2\varepsilon_f\varepsilon_g{(-3)}^{\frac{s+k_f+k_g-3}{2}}   &&(\alpha,\beta)\neq (0,0) ~{\rm and}~ f^*(\alpha)+g^*(\beta)+\lambda=0,\\
        &2\cdot3^{s-2}   &&\mu =0, (\alpha,\beta)\neq(0,0), f^*(\alpha)+g^*(\beta)+\lambda=1,\\
        &2\cdot 3^{s-2}+4\varepsilon_f\varepsilon_g{(-3)}^{\frac{s+k_f+k_g-3}{2}} && \mu =0, (\alpha,\beta)\neq(0,0), f^*(\alpha)+g^*(\beta)+\lambda=2,\\
        %&2\cdot 3^{s-2}+2\varepsilon_f\varepsilon_g{(-3)}^{\frac{s+k_f+k_g-3}{2}} &&\makecell[l]{c\neq0,f^*(\alpha)+g^*(\beta)+a=0,\alpha \neq 0 ,\beta=0\\{\rm \ or \ }c\neq0, f^*(\alpha)+g^*(\beta)+a=0,\alpha=0, \beta\neq0,}\\
        &2\cdot 3^{s-2}-\varepsilon_f\varepsilon_g{(-3)}^{\frac{s+k_f+k_g-1}{2}} &&\mu \neq 0, (\alpha,\beta)\neq(0,0), f^*(\alpha)+g^*(\beta)+\lambda=1,\\
        &2\cdot 3^{s-2}+\varepsilon_f\varepsilon_g{(-3)}^{\frac{s+k_f+k_g-3}{2}} &&\mu \neq 0, (\alpha,\beta)\neq(0,0), f^*(\alpha)+g^*(\beta)+\lambda=2,\\
        &3^{s-1}-\varepsilon_f\varepsilon_g{(-3)}^{\frac{s+k_f+k_g-1}{2}}  &&\mu \neq0~{\rm and}~(\alpha,\beta)=(0,0).\\
        %&2\cdot3^{s-2}+2\varepsilon_f\varepsilon_g{(-3)}^{\frac{s+k_f+k_g-3}{2}} && (\alpha, \beta) \notin\widetilde{\mathcal{S}\mathcal{R}}_f\times\widetilde{\mathcal{S}\mathcal{R}}_g \backslash \{(0,0)\}.\\   
      \end{aligned}
        \right.
        $$

(2) Since $\lambda=2$,  
  then by the definition of $n_{fg}(0)$ in Equation (\ref{eq.n_fg(0)}) and Lemma \ref{S_1}, 
  we have $n_{fg}(0)=3^{s-1}+\varepsilon_f\varepsilon_g(-3)^{\frac{s+k_f+k_g-1}{2}}$ when $s+k_f+k_g$ is odd.
  It clear the dimension of $\overline{\mathcal{C}_{D_{fg}(0)}}$ is $s+1$, 
  then this is a $[3^{s-1}+\varepsilon_f\varepsilon_g(-3)^{\frac{s+k_f+k_g-1}{2}},s+1]$ code.
  It again follows from the definition of ${\rm wt}(\overline{{\bf{c}}}_{fg}(0))$ and Lemmas \ref{S_1}, \ref{S_2} and \ref{S_3} that  
if $(\alpha, \beta) \notin\widetilde{\mathcal{S}\mathcal{R}}_f\times\widetilde{\mathcal{S}\mathcal{R}}_g \backslash \{(0,0)\}$,
then ${\rm wt}(\overline{{\bf{c}}}_{fg}(0))=2\cdot3^{s-2}-2\varepsilon_f\varepsilon_g{(-3)}^{\frac{s+k_f+k_g-3}{2}}$; 
and if $(\alpha, \beta) \in\widetilde{\mathcal{S}\mathcal{R}}_f\times\widetilde{\mathcal{S}\mathcal{R}}_g\bigcup\{(0,0)\}$, then 
  $$ {\rm wt}(\overline{{\bf{c}}}_{fg}(0))=\left\{
        \begin{aligned}
        &0   &&\mu =0~{\rm and}~(\alpha,\beta)=(0,0), \\    
        &2\cdot3^{s-2}-2\varepsilon_f\varepsilon_g{{(-3)}}^{\frac{s+k_f+k_g-3}{2}}   &&
        %\makecell[l]{c=0,f^*(\alpha)+g^*(\beta)+a=0\\ {\rm \ or \ }c\neq0,f^*(\alpha)+g^*(\beta)+a=0,\alpha \neq 0 ,\beta=0\\{\rm \ or \ }c\neq0, f^*(\alpha)+g^*(\beta)+a=0,\alpha=0, \beta\neq0,}\\
        (\alpha,\beta)\neq (0,0)~{\rm and}~ f^*(\alpha)+g^*(\beta)+\lambda=0,\\
        &2\cdot3^{s-2}-4\varepsilon_f\varepsilon_g{{(-3)}}^{\frac{s+k_f+k_g-3}{2}}   &&\mu =0, (\alpha,\beta)\neq(0,0), f^*(\alpha)+g^*(\beta)+\lambda=1,\\
        &2\cdot 3^{s-2} &&\mu =0, (\alpha,\beta)\neq(0,0), f^*(\alpha)+g^*(\beta)+\lambda=2,\\
        %&2\cdot 3^{s-2}-2\varepsilon_f\varepsilon_g{{(-3)}}^{\frac{s+k_f+k_g-3}{2}} &&\makecell[l]{c\neq0,f^*(\alpha)+g^*(\beta)+a=0,\alpha \neq 0 ,\beta=0\\{\rm \ or \ }c\neq0, f^*(\alpha)+g^*(\beta)+a=0,\alpha=0, \beta\neq0,}\\
        &2\cdot 3^{s-2}-\varepsilon_f\varepsilon_g{{(-3)}}^{\frac{s+k_f+k_g-3}{2}} &&\mu \neq 0, (\alpha,\beta)\neq(0,0), f^*(\alpha)+g^*(\beta)+\lambda=1,\\
        &2\cdot 3^{s-2}+\varepsilon_f\varepsilon_g{{(-3)}}^{\frac{s+k_f+k_g-1}{2}} &&\mu \neq 0, (\alpha,\beta)\neq(0,0), f^*(\alpha)+g^*(\beta)+\lambda=2,\\
        &3^{s-1}+\varepsilon_f\varepsilon_g{{(-3)}}^{\frac{s+k_f+k_g-1}{2}}   &&\mu \neq0~{\rm and}~(\alpha,\beta)=(0,0).\\
        %&2\cdot3^{s-2}-2\varepsilon_f\varepsilon_g{{(-3)}}^{\frac{s+k_f+k_g-3}{2}} && (\alpha, \beta) \notin\widetilde{\mathcal{S}\mathcal{R}}_f\times\widetilde{\mathcal{S}\mathcal{R}}_g \backslash \{(0,0)\}.\\    
      \end{aligned}
        \right.
        $$
        The weight distributions of the codes are also straightly derived from Lemma \ref{P_{f*,g*+c=0}} and \ref{P_{f*+g*+c}}.
        
        Since $s+k_f+k_g$ is odd and $k_f+k_g<s-3$, then $s-2>1$ and $\frac{s+k_f+k_g-3}{2}\geq 1$ is an integer. 
    It implies that $3\mid {\rm wt}(\overline{{\bf{c}}}_{fg}(0))$ for any ${\bf c}\in \overline{\mathcal{C}_{D_{fg}(0)}}$. 
    Then it turns out from Lemma \ref{lem.SO} that $\overline{\mathcal{C}_{D_{fg}(0)}}$ is self-orthogonal and we complete the proof.  
      \end{proof}

      \begin{remark}\label{rem.1} 
        Note that if we take $m=k_f=k_g=0$ and $\varepsilon_g=1$, then Theorem \ref{Th. double even ternary} can yield an infinite family of ternary self-orthogonal 
        $[3^{n-1}+\varepsilon_f(-3)^{\frac{n-2}{2}},n+1]$ codes, which coincide with \cite[Theorem 1]{Ternary & bent functions}. 
        Hence, Theorem \ref{Th. double even ternary} can be seen as a generalization of \cite{Ternary & bent functions} and  
        it can produce ternary self-orthogonal codes with more flexible parameters.  
        In addition, it is easy to see that Theorem \ref{Th. double odd ternary} can also similarly yield new infinite families of ternary self-orthogonal codes with flexible parameters. 
      \end{remark}

      \begin{table}[H]
        % table caption is above the table
        \caption{The weight distribution of $\overline{\mathcal{C}_{D_{fg}(0)}}$ when $k_f+k_g$ is odd and $\lambda=1$.} 
        \label{tab:double odd 1}       % Give a unique label
        % For LaTeX tables use
        \centering
        
        \begin{tabular}{clllc}
                \hline
             &Weight $i$ & Multiplicity $A_i$ \\
                \hline
                \hline
                 &$0$ & $1$ \\
                 &$2\cdot3^{s-2}+2\varepsilon_f\varepsilon_g{(-3)}^{\frac{s+k_f+k_g-3}{2}}$ & $3^{s+1}-2\cdot3^{s-k_f-k_g}+(-1)^{s}\varepsilon_f\varepsilon_g{(-3)}^{\frac{s-k_f-k_g+1}{2}}$ \\
                 &$2\cdot3^{s-2}$ & $3^{s-k_f-k_g-1}-1$ \\
                 &$2\cdot 3^{s-2}+4\varepsilon_f\varepsilon_g{(-3)}^{\frac{s+k_f+k_g-3}{2}}$ & $3^{s-k_f-k_g-1}+(-1)^{s}\varepsilon_f\varepsilon_g{(-3)}^{\frac{s-k_f-k_g-1}{2}}$\\ 
                 %&$2\cdot 3^{s-2}+\frac{2}{9}\varepsilon_f\varepsilon_g{\sqrt{3^*}}^{s+k_f+k_g+1}$ & $2\left(3^{s-k_f-k_g-1}-\varepsilon_f\varepsilon_g{\sqrt{3^*}}^{s-k_f-k_g-1}\right)$\\
                 &$2\cdot 3^{s-2}-\varepsilon_f\varepsilon_g{(-3)}^{\frac{s+k_f+k_g-1}{2}}$ & $2 \cdot 3^{s-k_f-k_g-1}-2$\\
                 &$2\cdot 3^{s-2}+\varepsilon_f\varepsilon_g{(-3)}^{\frac{s+k_f+k_g-3}{2}}$ & $2\cdot 3^{s-k_f-k_g-1}+2 (-1)^{s}\varepsilon_f\varepsilon_g{(-3)}^{\frac{s-k_f-k_g-1}{2}}$\\
                 &$3^{s-1}-\varepsilon_f\varepsilon_g{(-3)}^{\frac{s+k_f+k_g-1}{2}}$ &$2$\\
                 %&$2\cdot 3^{s-2}+\frac{2}{9}\varepsilon_f\varepsilon_g{\sqrt{3^*}}^{s+k_f+k_g+1}$ &$p(p^{s}-p^{s-k_f-k_g}-1)$\\
                 \hline       
                \end{tabular}
          \end{table}

          \begin{table}[H]
            % table caption is above the table
            \caption{The weight distribution of $\overline{\mathcal{C}_{D_{fg}(0)}}$ when $k_f+k_g$ is odd and $\lambda=2$.} 
            \label{tab:double odd -1}       % Give a unique label
            % For LaTeX tables use
            \centering
            
            \begin{tabular}{clllc}
                    \hline
                 &Weight $i$ & Multiplicity $A_i$ \\
                    \hline
                    \hline
                     &$0$ & $1$ \\
                     &$2\cdot3^{s-2}-2\varepsilon_f\varepsilon_g{(-3)}^{\frac{s+k_f+k_g-3}{2}}$ & $3^{s+1}-2\cdot3^{s-k_f-k_g}-(-1)^{s}\varepsilon_f\varepsilon_g{(-3)}^{\frac{s-k_f-k_g+1}{2}}$ \\
                     &$2\cdot3^{s-2}-4\varepsilon_f\varepsilon_g{(-3)}^{\frac{s+k_f+k_g-3}{2}}$ & $3^{s-k_f-k_g-1}-(-1)^{s}\varepsilon_f\varepsilon_g{(-3)}^{\frac{s-k_f-k_g-1}{2}}$ \\
                     &$2\cdot 3^{s-2}$ & $3^{s-k_f-k_g-1}-1$\\ 
                     %&$2\cdot 3^{s-2}-2\varepsilon_f\varepsilon_g{(-3)}^{\frac{s+k_f+k_g-3}{2}}$ & $2\left(3^{s-k_f-k_g-1}+\varepsilon_f\varepsilon_g{\sqrt{3^*}}^{s-k_f-k_g-1}\right)$\\
                     &$2\cdot 3^{s-2}-\varepsilon_f\varepsilon_g{(-3)}^{\frac{s+k_f+k_g-3}{2}}$ & $2\cdot 3^{s-k_f-k_g-1}-2(-1)^{s}\varepsilon_f\varepsilon_g{(-3)}^{\frac{s-k_f-k_g-1}{2}}$\\
                     &$2\cdot 3^{s-2}+\varepsilon_f\varepsilon_g{(-3)}^{\frac{s+k_f+k_g-1}{2}}$ & $2\cdot 3^{s-k_f-k_g-1}-2$\\
                     &$3^{s-1}+\varepsilon_f\varepsilon_g{(-3)}^{\frac{s+k_f+k_g-1}{2}}$ &$2$\\
                     %&$2\cdot 3^{s-2}-\frac{2}{9}\varepsilon_f\varepsilon_g{\sqrt{3^*}}^{s+k_f+k_g+1}$ &$p(p^{s}-p^{s-k_f-k_g}-1)$\\
                     \hline      
                    \end{tabular}
                \end {table}

                    \begin{example}\label{4+2+1,3+0+-1, odd 1}
                      Let $n=4$ and $m=3$, then $s=n+m=7$.
                      Let $f(x)=Tr_{3}^{4}(2x^{92})$ and $g(x)=Tr_{3}^{3}(\xi x^4)$, 
                      where $\xi$ is a primitive element of $\F_{3^3}$ .
                      On the one hand, it can be checked that $f(x)$ is a ternary weakly regular 2-plateaued function with $\widetilde{\mathcal{R}}_f(\alpha )\in \{0,-3^3\zeta_3^{f^*(\alpha)}\}$ 
                      for all $\alpha \in \F_{3^4}$, where the sign of the Walsh transform of $f(x)$ is $\varepsilon_f=1$ and $f^*$ is an unbalanced function over $\F_{3^3}$ satisfying $f^*(0)=0$.
                      On the other hand, it follows that $g(x)$ is a ternary weakly regular 0-plateaued  and the sign of the Walsh transform of $g(x)$ is $\varepsilon_g=-1$.
                      Hence, it turns out from Theorem \ref{Th. double odd ternary} that for $\lambda=1$, $\overline{\mathcal{C}_{D_{fg}(0)}}$ is a ternary six-weight self-orthogonal [810,8,486] code 
                      and its weight enumerator is  $$1+80x^{486}+180x^{513}+6048x^{540}+160x^{567}+90x^{594}+2x^{810}.$$
                      Verified by Magma \cite{Magma}, these results are true.
                    \end{example}

              \begin{example}\label{3+1+-1,3+0+-1, odd -1}
                Let $n=3$ and $m=3$, then $s=m+n=6$.
                Let $f(x)=Tr_{3}^{3}(\xi^{22}x^{13}+\xi^7x^4+\xi x^2)$, where $\xi$ is a primitive element of $\F_{3^3}$, 
                then it follows that $f(x)$ is a ternary weakly regular 1-plateaued function with $\widetilde{\mathcal{R}}_f(\alpha )\in \{0,9\varepsilon_f\zeta_3^{f^*(\alpha)}\}=\{0,-9,-9\zeta_3,-9\zeta_3^2\}$ 
                for all $\alpha \in \F_{3^3}$, where the sign of the Walsh transform of $f(x)$ is $\varepsilon_f=-1$ and $f^*$ is an unbalanced function over $\F_{3^3}$ satisfying $f^*(0)=0$.
                Let $g(x)=Tr_{3}^{3}(\omega x^2)$, where $\omega$ is a primitive element of $\F_{3^3}$,
                then it can be checked that $g(x)$ is a ternary weakly regular 0-plateaued and the sign of the Walsh transform of $g(x)$ is $\varepsilon_g=-1$.
                Hence, it turns out from Theorem \ref{Th. double odd ternary} that for $\lambda=2$, $\overline{\mathcal{C}_{D_{fg}(0)}}$ is a ternary six-weight self-orthogonal [216,7,126] code 
                and its weight enumerator is $$1+72x^{126}+160x^{135}+1728x^{144}+144x^{153}+80x^{162}+2x^{216}.$$
                Verified by Magma \cite{Magma}, these results are true.
              \end{example}

              \subsection{Two new families of ternary self-orthogonal codes from $f \notin \mathcal{WRP}$ and $g \in \mathcal{WRP}$. }
              In this subsection, we let notations be the same as above and assume that $f(x)=Tr_{3}^{n}(x)$ and $g(y) \in \mathcal{WRP}$. 
              Let $s=n+m$ and denote by 
              \[\begin{split}
                D_{g}(0)&=\left\{(x,y)\in \F_{3^n}\times\F_{3^m}:Tr_{3}^{n}(x)+g(y)+\lambda=0\right\},
              \end{split}\] 
               where $\lambda \in \F_3^*$.
              Let $n_g(0)$ denote the length of the augmented codes $\overline{\mathcal{C}_{D_{g}(0)}}$ 
              and ${\rm wt}(\overline{{\bf{c}}}_{g}(0))$ denote the weight of nonzero codewords $\overline{{\bf{c}}}_g(0)$.
              It can be verified that 
    \begin{align}\label{eq.n_fg(0)}
      n_{g}(0)=\#\left\{(x,y)\in \F_{3^n}\times \F_{3^m}:Tr_3^n(x)+g(y)+\lambda=0\right\}
    \end{align}
    and then  
    \begin{align*}
      n_{g}(0)=&\frac{1}{3}\sum_{t\in \F_3}\sum_{(x,y)\in \F_{3^n}\times \F_{3^m}}\zeta_3^{t(Tr_{3}^{n}(x)+g(y)+\lambda)}
        =3^{s-1}.
    \end{align*}
                   Then we determine their weight distributions in the following theorems.
                    %{\rm wt}(\overline{{\bf{c}}}(0))=&n-N_{g}(\NSQ)\\
                  %=& \frac{1}{p}\sum_{t\in \F_p}\sum_{(x,y)\in \F_{p^n}\times \F_{p^m}}\zeta_p^{t(f(x)+g(y)+\lambda)}-\frac{1}{p^2}\sum_{t\in \F_p}\sum_{u \in \F_p}\sum_{(x,y)\in \F_{p^n}\times \F_{p^m}}\zeta_p^{t(f(x)+g(y)+\lambda)+u(Tr(\alpha x+\beta y)+\mu)}\\
                  %=& p^{s-1}+\frac{1}{p}\sum_{t\in \F_p^*}\sum_{(x,y)\in \F_{p^n}\times \F_{p^m}}\zeta_p^{t(f(x)+g(y)+\lambda)}
                  %-p^{s-2}-\frac{1}{p^2}\sum_{t\in \F_p^*}\sum_{(x,y)\in \F_{p^n}\times \F_{p^m}}\zeta_p^{t(f(x)+g(y)+\lambda)}\\
                  %-&\frac{1}{p^2}\sum_{u\in \F_p^*}\sum_{(x,y)\in \F_{p^n}\times \F_{p^m}}\zeta_p^{u(Tr(\alpha x+\beta y)+\mu )}-\frac{1}{p^2}\sum_{t\in \F_p}\sum_{u \in \F_p}\sum_{(x,y)\in \F_{p^n}\times \F_{p^m}}\zeta_p^{t(f(x)+g(y)+\lambda)+u(Tr(\alpha x+\beta y)+\mu)}\\
                  %=& (p-1)p^{s-2}+\frac{p-1}{p^2}S_1-\frac{1}{p^2}S_2-\frac{1}{p^2}S_3.

              \begin{theorem}\label{Th.simple even  0,SQ,NSQ}
              Let $g\in \mathcal{WRP}$ with $\widetilde{\mathcal{R}}_g(\beta)=\varepsilon_g{\sqrt{3^*}}^{m+k_g}\zeta_3^{g^*(\beta)}$ for every $\beta \in \widetilde{\mathcal{S}\mathcal{R}}_g$,
                where $\varepsilon_g\in \{ -1, 1\}$ and $0\leq k_g\leq m$.
                If $m+k_g$ is even and $\lambda \in \F_3^*$, then $\overline{\mathcal{C}_{D_{g}(0)}}$ is a four-weight self-orthogonal $[3^{s-1},s+1]$ ternary linear code with weight distribution given in Table \ref{tab:simple even 0,SQ,NSQ}. 
              \end{theorem}
              \begin{proof}
                 The proof is very similar to that of Theorem \ref{Th. double even ternary} and the main difference is that we use 
                 Lemmas \ref{S_5} and \ref{Ng(0)} here.  
                 %Besides, the weights of nonzero codewords of $\overline{\mathcal{C}_{D_{g}(0)}}$, $\overline{\mathcal{C}_{D_{g}(1)}}$ and $\overline{\mathcal{C}_{D_{g}(2)}}$  are all divided by 3, 
                 %So $\overline{\mathcal{C}_{D_{g}(0)}}$, $\overline{\mathcal{C}_{D_{g}(1)}}$ and $\overline{\mathcal{C}_{D_{g}(2)}}$ are ternary self-orthogonal codes.
                \end{proof}
                
                \begin{theorem}\label{Th.simple odd 0,SQ,NSQ}
                  Let $g\in \mathcal{WRP}$ with $\widetilde{\mathcal{R}}_g(\beta)=\varepsilon_g{\sqrt{3^*}}^{m+k_g}\zeta_p^{g^*(\beta)}$ for every $\beta \in \widetilde{\mathcal{S}\mathcal{R}}_g$,
                    where $\varepsilon_g\in \{ -1, 1\}$ and $0\leq k_g\leq m$.
                    If $m+k_g$ is odd and $\lambda \in \F_3^*$, then $\overline{\mathcal{C}_{D_{g}(0)}}$
                    is a four-weight self-orthogonal $[3^{s-1},s+1]$ ternary linear code with weight distribution given in Table \ref{tab:simple odd 0,SQ,NSQ}. 
                  \end{theorem}
                \begin{proof}
                     The proof is very similar to that of Theorem \ref{Th.simple even 0,SQ,NSQ} and we omit it here. 
                \end{proof}
              
                \begin{table}[H]
                  % table caption is above the table
                  \caption{The weight distribution of $\overline{\mathcal{C}_{D_{g}(0)}}$ when $m+k_g$ is even.} 
                  \label{tab:simple even 0,SQ,NSQ}       % Give a unique label
                  % For LaTeX tables use
                  \centering
                  
                  \begin{tabular}{clllc}
                          \hline
                       &Weight $i$ & Multiplicity $A_i$ \\
                          \hline
                          \hline
                           &$0$ & $1$ \\
                           &$2\cdot3^{s-2}-2\cdot3^{n-2}\varepsilon_g{(-3)}^{\frac{m+k_g}{2}}$ & $2\cdot3^{m-k_g}$ \\
                           &$2\cdot3^{s-2}+3^{n-2}\varepsilon_g{(-3)}^{\frac{m+k_g}{2}}$ & $4\cdot3^{m-k_g}$ \\
                           &$2\cdot3^{s-2}$ &$3^{s+1}-2\cdot3^{m-k_g+1}-3$\\
                           &$3^{s-1}$ & $2$\\
                           \hline      
                          \end{tabular}
                    \end{table}
                    
                    \begin{table}[H]
                      % table caption is above the table
                      \caption{The weight distribution of $\overline{\mathcal{C}_{D_{g}(0)}}$ when $m+k_g$ is odd.} 
                      \label{tab:simple odd 0,SQ,NSQ}       % Give a unique label
                      % For LaTeX tables use
                      \centering
                      
                      \begin{tabular}{clllc}
                              \hline
                           &Weight $i$ & Multiplicity $A_i$ \\
                              \hline
                              \hline
                               &$0$ & $1$ \\
                               &$2\cdot3^{s-2}-3^{n-2}\varepsilon_g{(-3)}^{\frac{m+k_g+1}{2}}$ & $2\cdot3^{m-k_g}$ \\
                               &$2\cdot3^{s-2}+3^{n-2}\varepsilon_g{(-3)}^{\frac{m+k_g+1}{2}}$ & $2\cdot3^{m-k_g}$ \\
                               &$2\cdot3^{s-2}$ &$3^{s+1}-4\cdot3^{m-k_g}-3$\\
                               &$3^{s-1}$ & $2$\\
                               \hline      
                              \end{tabular}
                        \end{table}
              
                        \begin{example}\label{1,4+0+-1, even}
                          Let $n=1$ and $m=4$, then $s=n+m=5$.
                          Let $g(x)=Tr_{3}^{4}(\xi x^2)$, where $\xi$ is a primitive element of $\F_{3^4}$,
                          it follows that $g(x)$ is a ternary weakly regular 0-plateaued function and the sign of the Walsh transform of $g(x)$ is $\varepsilon_g=-1$.
                          Therefore, it turns out from Theorem \ref{Th.simple even 0,SQ,NSQ} that $\overline{\mathcal{C}_{D_{g}(0)}}$ is a ternary four-weight self-orthogonal [81,6,51] code 
                          and its weight enumerator is  $$1+324x^{51}+240x^{54}+162x^{60}+2x^{81},$$
                          which is optimal according to \cite{Codetables}. 
                          Verified by Magma \cite{Magma}, these results are true. 
                        \end{example}

                        \begin{example}\label{1,3+0+-1, odd}
                        Let $n=1$ and $m=3$, then $s=n+m=4$.  
                        Let $g(x)=Tr_{3}^{3}(\xi x^4)$, where $\xi$ is a primitive element of $\F_{3^3}$,
                        it follows that $g(x)$ is a ternary weakly regular 0-plateaued function and the sign of the Walsh transform of $g(x)$ is $\varepsilon_g=-1$.
                        Therefore, it  turns out from Theorem \ref{Th.simple odd 0,SQ,NSQ} that $\overline{\mathcal{C}_{D_{g}(1)}}$  is a ternary four-weight self-orthogonal [27,5,15] code 
                        and its  weight enumerator is  $$1+54x^{15}+132x^{18}+54x^{21}+2x^{27},$$
                        which is almost optimal since the optimal code with length 27 and dimension 5 has minimum weight 16 by \cite{Codetables}. 
                        Verified by Magma \cite{Magma}, these results are true. 
                        \end{example}

    \section{Dual codes of several families of ternary self-orthogonal codes}\label{sec.Dual SO}

    In this section, we study dual codes of these ternary self-orthogonal linear codes constructed in Section \ref{sec.SO}. 
    Let $\overline{\mathcal{C}_{D_{fg}(0)}}^\bot$ and $\overline{\mathcal{C}_{D_{g}(0)}}^\bot$ 
    denote the dual codes of $\overline{\mathcal{C}_{D_{fg}(0)}}$ and  $\overline{\mathcal{C}_{D_{g}(0)}}$, respectively.  
    Next, we determine their parameters.

    \begin{theorem}\label{Th.double dual}
      Let $s=n+m$, where $n$ and $m$ are two positive integers.
      Let $f,g\in \mathcal{WRP}$ with $\widetilde{\mathcal{R}}_f(\alpha )=\varepsilon_f{\sqrt{p^*}}^{n+k_f}\zeta_p^{f^*(\alpha)}$ for every $\alpha \in \widetilde{\mathcal{S}\mathcal{R}}_f$
     and $\widetilde{\mathcal{R}}_g(\beta)=\varepsilon_g{\sqrt{p^*}}^{m+k_g}\zeta_p^{g^*(\beta)}$ for every $\beta \in \widetilde{\mathcal{S}\mathcal{R}}_g$,
      where $\varepsilon_f,\varepsilon_g\in \{ -1, 1\}$, $0\leq k_f\leq n$ and $0\leq k_g\leq m$.
     Then the following statements hold. 
       \begin{enumerate}
         \item [\rm (1)] If $s+k_f+k_g$ is even, $k_f+k_g<s-3$ and $\lambda \in \F_3^*$, then $\overline{\mathcal{C}_{D_{fg}(0)}}^\bot$ is a ternary $[3^{s-1}+\varepsilon_f\varepsilon_g{(-3)}^{\frac{s+k_f+k_g-2}{2}},3^{s-1}+\varepsilon_f\varepsilon_g{(-3)}^{\frac{s+k_f+k_g-2}{2}}-s-1,3]$ linear code.
         
         \item [\rm (2)] If $s+k_f+k_g$ is odd, $k_f+k_g<s-3$ and $\lambda=1$, then $\overline{\mathcal{C}_{D_{fg}(0)}}^\bot$ is a ternary $[3^{s-1}-\varepsilon_f\varepsilon_g{(-3)}^{\frac{s+k_f+k_g-1}{2}},3^{s-1}-\varepsilon_f\varepsilon_g{(-3)}^{\frac{s+k_f+k_g-1}{2}}-s-1,3]$ linear code.
   
         \item [\rm (3)] If $s+k_f+k_g$ is odd, $k_f+k_g<s-3$ and $\lambda=2$, then $\overline{\mathcal{C}_{D_{fg}(0)}}^\bot$ is a ternary $[3^{s-1}+\varepsilon_f\varepsilon_g{(-3)}^{\frac{s+k_f+k_g-1}{2}},3^{s-1}+\varepsilon_f\varepsilon_g{(-3)}^{\frac{s+k_f+k_g-1}{2}}-s-1,3]$ linear code.
       \end{enumerate}
        \end{theorem}             
       \begin{proof}
         (1) Let $d(\overline{\mathcal{C}_{D_{fg}(0)}}^{\bot})$ denote the minimum weight of $\overline{\mathcal{C}_{D_{fg}(0)}}^\bot$.
         It follows from the definition of $\overline{\mathcal{C}_{D_{fg}(0)}}$ that $d(\overline{\mathcal{C}_{D_{fg}(0)}}^{\bot}) \geq 2$.
         Let $(1,A_1,A_2,\cdots,A_n)$ and $(1,A_1^\bot,A_2^\bot,\cdots,A_n^\bot)$ denote the weight distributions of $\overline{\mathcal{C}_{D_{fg}(0)}}$ and $\overline{\mathcal{C}_{D_{fg}(0)}}^\bot$, respectively.
         If $s+k_f+k_g$ is even,  it follows from the third Pless power moment $(P_3)$ that 
   \begin{align}\label{eq.double third Pless power moment}
     \sum_{j = 0}^{n}j^2A_j&=3^{k-2}\left(2n(3n-n+1)-(6n-3-2n+2)A_1^\bot+2A_2^\bot\right).
    \end{align}
    Substituting the weight distribution of $\overline{\mathcal{C}_{D_{fg}(0)}}$ we obtain in Theorem \ref{Th. double even ternary} 
    into Equation (\ref{eq.double third Pless power moment}), we immediately have 
   \begin{align*}
     A_2^\bot=0.
   \end{align*}    
     Note also that the fourth Pless power moment ($P_4$) yields that  
     \begin{align}\label{eq.double fourth Pless power moment}
     \sum_{j = 0}^{n}j^3A_j&=3^{k-3}[2n(3^2n^2-6n^2+9n-3+n^2-3n+2)-6A_3^\bot].  
   \end{align}
   By combining the weight distribution of $\overline{\mathcal{C}_{D_{fg}(0)}}$ we present in Theorem \ref{Th. double even ternary} with  Equation (\ref{eq.double fourth Pless power moment}), 
   we get 
   \begin{align*}
     A_3^\bot=&\left(10\cdot 3^{2s-2}-16\cdot3^{s-2}-32\cdot3^{2s-k_f-k_g-3}-3^{-1}\right)(-3)^{\frac{s+k_f+k_g-2}{2}}\\
     &-88\cdot3^{3s-k_f-k_g-6}-4\cdot3^{2s+2}-8\cdot 3^{2s-k_f-k_g-3}+20\cdot3^{2s-5}>0 
   \end{align*}
   and hence, $d(\overline{\mathcal{C}_{D_{fg}(0)}}^{\bot})=3$. This completes the desired result (1). 
   
   (2) and (3) By taking an argument completely similar to that of (1) above, 
   the desired results (2) and (3) hold. % the minimum weight $d(\overline{\mathcal{C}_{D_{fg}(0)}}^{\bot})$ of $\overline{\mathcal{C}_{D_{fg}(0)}}^\bot$ is 3 when $s+k_f+k_g$ is odd and $\lambda=\pm 1$.
   \end{proof}  
    
   \begin{theorem}\label{Th.simple dual}
    Let $s=n+m$, where $n$ and $m$ are two positive integers.
    Let $\lambda \in \F_3^*$ and $g\in \mathcal{WRP}$ with $\widetilde{\mathcal{R}}_g(\beta)=\varepsilon_g{\sqrt{p^*}}^{m+k_g}\zeta_p^{g^*(\beta)}$ 
    for every $\beta \in \widetilde{\mathcal{S}\mathcal{R}}_g$, where $\varepsilon_g\in \{ -1, 1\}$ and $0\leq k_g\leq m$.   
    Then the following statements hold. 
    \begin{enumerate}
        \item [\rm (1)] If $m+k_g$ is even, then $\overline{\mathcal{C}_{D_{g}(0)}}^\bot$ is a ternary $[3^{s-1},3^{s-1}-s-1,3]$ linear code.
        
        \item [\rm (2)] If $m+k_g$ is odd, then $\overline{\mathcal{C}_{D_{g}(0)}}^\bot$ is a ternary $[3^{s-1},3^{s-1}-s-1,3]$ linear code. 
     \end{enumerate}
       \end{theorem}             
      \begin{proof}
        %We only consider the minimum weight of $\overline{\mathcal{C}_{D_{g}(0)}}^\bot$ here and the other codes are similar. 
        Let $d^\bot$ denote the minimum weight of $\overline{\mathcal{C}_{D_{g}(0)}}^\bot$
        and it follows from the definition of $\overline{\mathcal{C}_{D_{g}(0)}}^\bot$ that $d^\bot \geq 2$.
        Let $(1,A_1,A_2,\cdots,A_n)$ and $(1,A_1^\bot,A_2^\bot,\cdots,A_n^\bot)$ denote the weight distributions of $\overline{\mathcal{C}_{D_{g}(0)}}$ and $\overline{\mathcal{C}_{D_{g}(0)}}^\bot$, respectively.
        
        For even $m+k_g$, it follows from the third Pless power moment ($P_3$) that 
  \begin{align}\label{eq.simple third Pless power moment}
    \sum_{j = 0}^{n}j^2A_j&=3^{k-2}\left(2n(3n-n+1)-(6n-3-2n+2)A_1^\bot+2A_2^\bot\right).
   \end{align}
   Substituting the  weight distribution of $\overline{\mathcal{C}_{D_{g}(0)}}$ we obtain in Theorem \ref{Th. double even ternary} into Equation (\ref{eq.simple third Pless power moment}), 
   we have 
  \begin{align*}
    A_2^\bot=0.
  \end{align*}    
    From the fourth Pless power moment, we also derive 
    \begin{align}\label{eq.simple fourth Pless power moment}
    \sum_{j = 0}^{n}j^3A_j&=3^{k-3}[2n(3^2n^2-6n^2+9n-3+n^2-3n+2)-6A_3^\bot].  
  \end{align}
  Substituting the weight distribution of $\overline{\mathcal{C}_{D_{g}(0)}}$ we obtain in Theorem \ref{Th. double even ternary} into Equation (\ref{eq.simple fourth Pless power moment}), then we obtain
  \begin{align*}
    A_3^\bot=3^{s-2}\left(3^{n-2}\left(2 \cdot \varepsilon_g(-3)^\frac{m+k_g}{2}+3^m\right)-1\right)>0,
  \end{align*}
which further deduces that $d^\bot=3$. 
  
  For odd $m+k_g$, it follows from Theorem \ref{Th.simple odd 0,SQ,NSQ} and Equation (\ref{eq.simple third Pless power moment}) that $$A_2^\bot=0.$$ 
  With Theorem \ref{Th.simple odd 0,SQ,NSQ} and Equation (\ref{eq.simple fourth Pless power moment}) again, we  have 
     \begin{align*}
      A_3^\bot=3^{s-2}\left(8\cdot3^{2s-2}+3^{s-2}-1\right)>0
    \end{align*}
  and hence, $d^\bot=3$. 
  This completes the proof. 
  \end{proof}

\begin{remark}
  Although the dual codes $\overline{\mathcal{C}_{D_{g}(0)}}^\bot$ in Parts (1) and (2)  of Theorem \ref{Th.simple dual} have the same parameters, 
  they must be not equivalent since their original codes $\overline{\mathcal{C}_{D_{g}(0)}}$ have different weight distributions according to Tables \ref{tab:simple even 0,SQ,NSQ} and Table \ref{tab:simple odd 0,SQ,NSQ}.   
  %Similarly, the dual codes $\overline{\mathcal{C}_{D_{g}(0)}}^\bot$ in Parts (1) and (2) of Theorem \ref{Th.simple dual} are not equivalent either.   
\end{remark}

\section{Application to ternary LCD codes}\label{sec.LCD codes}

In this section, we apply our new ternary self-orthogonal codes constructed in Section \ref{sec.SO} to obtain new families of ternary LCD codes. 
To this end, we need to recall some useful definitions and results as follows. 

Let $I$ be an {\em identity matrix} and $O$ be a {\em zero matrix}. 
For any matrix $G$, if $GG^T=I$, we call $G$ a {\em row-orthogonal matrix}; 
and if $GG^T=O$, we call $G$ a {\em row-self-orthogonal matrix}. 
A {\em leading-systematic linear code} is referred as a linear code generated by a matrix of the form $G=(I \ A)$. 
Moreover, the matrix $G$ is called {\em the systematic generator matrix} of the code. 
The relationship between LCD codes and row-orthogonal matrices has been proposed in \cite{jg.LCD}. 
\begin{lemma}[\cite{jg.LCD}]\label{lm.jg.LCD}
    A leading-systematic linear code $\mathcal{C}$ is an LCD code if its systematic generator matrix $G=(I \ A)$ is row-orthogonal, 
    $i.e.,$ the matrix $A$ is row-self orthogonal.  
\end{lemma}

%Above all, we can obtain leading-systematic LCD codes with generator matrix $(I \ A)$, where $A$ is a generator matrix of  self-orthogonal code $\mathcal{C}$.

\begin{lemma}\label{lm.generator matrix}
  Let $p$ be an odd prime and $\xi$ be a generator element of $\F_{p^{s}}^*=\langle\xi\rangle$.
  For any $p$-ary linear augmented code $\overline{\mathcal{C}_{D}}$ defined in Equation (\ref{AC_D}) with the defining set $D$, 
  where $D=\{d_1,d_2,\cdots,d_t\}$, 
  a generator matrix of $\overline{\mathcal{C}_{D}}$ is given by 
  \begin{equation}\label{eq.generator matrix}
    G=\left(\begin{array}{cccc} 
    Tr_{p}^{s}(d_1) & Tr_{p}^{s}(d_2) &\cdots & Tr_{p}^{s}(d_t)  \\
    Tr_{p}^{s}(\xi d_1) & Tr_{p}^{s}(\xi d_2) &\cdots & Tr_{p}^{s}(\xi d_t)  \\
    \vdots & \vdots & \ddots  & \vdots \\
    Tr_{p}^{s}(\xi^{s-1}d_1) & Tr_{p}^{s}(\xi^{s-1}d_2) &\cdots & Tr_{p}^{s}(\xi^{s-1}d_t)  \\
    1 & 1 & \cdots & 1 \\
    \end{array}\right).
    \end{equation}
\end{lemma}
\begin{proof}
  The generator matrix $G$ follows from the definition of the augmented code $\overline{\mathcal{C}_{D}}$ in Equation (\ref{AC_D}) 
  and the fact that $\{1,\xi,\cdots,\xi^{s-1}\}$ forms a basis of $\F_{p^s}$ over $\F_p$. 
\end{proof}

\begin{theorem}\label{Th. LCD even ternary}
  Let $s=n+m$, where $n$ and $m$ are two positive integers.
  Let $f,g\in \mathcal{WRP}$ with $\widetilde{\mathcal{R}}_f(\alpha )=\varepsilon_f{\sqrt{3^*}}^{n+k_f}\zeta_3^{f^*(\alpha)}$ for every $\alpha \in \widetilde{\mathcal{S}\mathcal{R}}_f$
  and $\widetilde{\mathcal{R}}_g(\beta)=\varepsilon_g{\sqrt{3^*}}^{m+k_g}\zeta_3^{g^*(\beta)}$ for every $\beta \in \widetilde{\mathcal{S}\mathcal{R}}_g$,
   where $\varepsilon_f,\varepsilon_g\in \{ -1, 1\}$, $0\leq k_f\leq n$ and $0\leq k_g\leq m$.
    Assume that $s+k_f+k_g$ is even and $k_f+k_g<s-3$.
    Let $\lambda \in \F_3^*$ and $\overline{\mathcal{C}_{D_{fg}(0)}}$ be defined in Equation (\ref{AC_D}),
    then the matrix $\overline{G}=(I \ G)$ generates a ternary LCD $[3^{s-1}+\varepsilon_f\varepsilon_g(-3)^{\frac{s+k_f+k_g-2}{2}}+s+1,s+1,d]$ LCD code $\mathcal{C}$, 
    where $d\geq \min\left\{1+2\cdot3^{s-2},1+2\cdot3^{s-2}+2\varepsilon_f\varepsilon_g(-3)^{\frac{s+k_f+k_g-2}{2}}\right\}$ 
    and $G$ is the generator matrix of $\overline{\mathcal{C}_{D_{fg}(0)}}$.  
    Besides, $\mathcal{C}^{\bot}$ is a ternary LCD  $[3^{s-1}+\varepsilon_f\varepsilon_g(-3)^{\frac{s+k_f+k_g-2}{2}}+s+1,3^{s-1}+\varepsilon_f\varepsilon_g(-3)^{\frac{s+k_f+k_g-2}{2}},3]$ code 
  and it is at least almost optimal with respect to the Sphere-packing bound given in Equation (\ref{eq.Sphere-packing}).    
  \end{theorem}
  \begin{proof}
  According to Theorem \ref{Th. double even ternary} and Lemma \ref{lm.jg.LCD}, the matrix $\overline{G}=(I \ G)$ generates a ternary LCD code $\mathcal{C}$, 
  where  

  \begin{equation}\label{eq.G.generator matrix}
    \overline{G}=\left(\begin{array}{ccccccccc} 
    1&0& \cdots&0&0& Tr_{3}^{s}(d_1) & Tr_{3}^{s}(d_2) &\cdots & Tr_{3}^{s}(d_t)  \\
    0&1& \cdots&0&0&Tr_{3}^{s}(\xi d_1) & Tr_{3}^{s}(\xi d_2) &\cdots  & Tr_{3}^{s}(\xi d_t)  \\
    \vdots&\vdots& \ddots&\vdots&\vdots&\vdots & \vdots & \ddots  & \vdots \\
    1&0& \cdots&1&0&Tr_{3}^{s}(\xi^{s-1}d_1) & Tr_{3}^{s}(\xi^{s-1}d_2) &\cdots & Tr_{3}^{s}(\xi^{s-1}d_t)  \\
    1&0& \cdots&0&1&1 & 1 & \cdots & 1 \\
    \end{array}\right)
    \end{equation}
    and $D_{fg}(0)=\{d_1,d_2,\cdots,d_t\}$. 
    Note also that a dual code of an LCD code is still LCD. Hence, $\mathcal{C}^{\bot}$ is also a ternary LCD code. 
    Let $d(\overline{\mathcal{C}_{D_{fg}(0)}})$ denote the minimum weight of $\overline{\mathcal{C}_{D_{fg}(0)}}$ and $d$ denote the minimum weight  of $\mathcal{C}$. 
    Then it is clear that $d\geq d(\overline{\mathcal{C}_{D_{fg}(0)}})+1$. 
    Moreover, it turns out from Table \ref{tab: double even } that $d\geq \min\left\{1+2\cdot3^{s-2},1+2\cdot3^{s-2}+2\varepsilon_f\varepsilon_g(-3)^{\frac{s+k_f+k_g-2}{2}}\right\}$. 
    
    By Theorem \ref{Th.double dual}, the minimum weight of the dual code $\overline{\mathcal{C}_{D_{fg}(0)}}^{\bot}$ is $3$, 
    which implies that any two columns of $G$ are linearly independent and there exist three dependent columns of $G$ over $\F_3$.
    Let $d^{\bot}$ denote the minimum weight of $\mathcal{C}^\bot$. Now, we  prove that  $d^{\bot}=3$.  
    Then it suffices to prove that the column vector $(0,0,\cdots,0,1)^T$ and any column of $G$ are linearly independent over $\F_3$. 
    Suppose that $(0,0,\cdots,0,1)^T$ and $(Tr_{3}^{s}(d_i),Tr_{3}^{s}(\xi d_i),\cdots,Tr_{3}^{s}(\xi^{s-1}d_i),1)^T$ are linearly dependent over $\F_3$
    for some $d_i\in D_{fg}(0)$, 
    then we have 
   % $$ \left\{
    %\begin{aligned}
    %& \left( Tr(\xi^0 d_i)=0, \\
    %& \left( Tr(\xi d_i)=0, \\
    %& \left(  \vdots,\\  
    %& \left( Tr(\xi^{s-1}d_i)=0 \
  %\end{aligned} 
  %\right.
  %$$
  $$ \left\{
    \begin{array}{l}
      Tr_{3}^{s}( d_i)= 0, \\
      Tr_{3}^{s}(\xi d_i)= 0, \\
       ~~~~~~~\vdots\\
      Tr_{3}^{s}(\xi^{s-1}d_i)= 0. \\ 
    \end{array}
  \right.
  $$
  For any $x=\sum_{i-0}^{s-1}k_i\xi^i \in \F_{3^{s}}$ and $k_i \in \F_3$,
  it then deduces that $$Tr_{3}^{s}(x d_i)=\sum_{i=0}^{s-1}k_iTr_{3}^{s}(\xi^i d_i)=0.$$
  Note that $Tr_{3}^{s}(x)$ is a $3$-ary function from $\F_{3^{s}}$ to $\F_3$ and $\vert {\rm Ker}(Tr_{3}^{s})\vert=3^{s-1}$ since $0 \notin D_{fg}(0)$ for any $\lambda \in \F_3^*$. 
  This produces a contradiction and hence, any two columns of $\overline{G}$ are linearly independent over $\F_3$.  
  That is $d^{\bot}=3$. 
  
  Moreover, for any fixed length $3^{s-1}+\varepsilon_f\varepsilon_g(-3)^{\frac{s+k_f+k_g-2}{2}}+s+1$ and dimension $3^{s-1}+\varepsilon_f\varepsilon_g(-3)^{\frac{s+k_f+k_g-2}{2}}$, 
  it follows from the Sphere-packing bound given in Equation (\ref{eq.Sphere-packing}) that $d^\bot$ is at most $4$. 
  Hence, $\mathcal{C}^\bot$ is at least almost optimal with respect to the sphere-packing bound.     
  
  We have completed the whole proof. 
\end{proof}

\begin{theorem}\label{Th. LCD odd ternary}
  Let $s=n+m$, where $n$ and $m$ are two positive integers.
  Let $f,g\in \mathcal{WRP}$ with $\widetilde{\mathcal{R}}_f(\alpha )=\varepsilon_f{\sqrt{3^*}}^{n+k_f}\zeta_3^{f^*(\alpha)}$ for every $\alpha \in \widetilde{\mathcal{S}\mathcal{R}}_f$
  and $\widetilde{\mathcal{R}}_g(\beta)=\varepsilon_g{\sqrt{3^*}}^{m+k_g}\zeta_3^{g^*(\beta)}$ for every $\beta \in \widetilde{\mathcal{S}\mathcal{R}}_g$,
   where $\varepsilon_f,\varepsilon_g\in \{ -1, 1\}$, $0\leq k_f\leq n$ and $0\leq k_g\leq m$.
    Assume that $s+k_f+k_g$ is odd, $k_f+k_g<s-3$ and $\lambda \in \F_3^*$, 
    $\overline{\mathcal{C}_{D_{fg}(0)}}$ be defined in Equation (\ref{AC_D}) with a generator matrix $G$. 
    Then the following statements hold.  
\begin{enumerate}
  \item [\rm (1)] If $\lambda=1$, then the matrix $\overline{G}=(I\ G)$ generates a ternary LCD  $[3^{s-1}-\varepsilon_f\varepsilon_g(-3)^{\frac{s+k_f+k_g-1}{2}}+s+1,s+1,d]$ code $\mathcal{C}$, 
  where $d\geq \min\left\{1+2\cdot3^{s-2},1+2\cdot3^{s-2}+4\varepsilon_f\varepsilon_g(-3)^{\frac{s+k_f+k_g-3}{2}}\right\}$.
  Besides, $\mathcal{C}^{\bot}$ is a ternary LCD $[3^{s-1}-\varepsilon_f\varepsilon_g(-3)^{\frac{s+k_f+k_g-1}{2}}+s+1,3^{s-1}-\varepsilon_f\varepsilon_g(-3)^{\frac{s+k_f+k_g-1}{2}},3]$ code 
  and it is at least almost optimal with respect to the sphere-packing bound given in Equation (\ref{eq.Sphere-packing}). 
  \item [\rm (2)] If $\lambda=2$, then the matrix $\overline{G}=(I\ G)$ generates a ternary LCD $[3^{s-1}+\varepsilon_f\varepsilon_g(-3)^{\frac{s+k_f+k_g-1}{2}}+s+1,s+1,d]$ code $\mathcal{C}$, 
  where $d\geq \min\left\{1+2\cdot3^{s-2},1+2\cdot3^{s-2}-4\varepsilon_f\varepsilon_g(-3)^{\frac{s+k_f+k_g-3}{2}}\right\}$. 
  Besides, $\mathcal{C}^{\bot}$ is a ternary LCD $[3^{s-1}+\varepsilon_f\varepsilon_g(-3)^{\frac{s+k_f+k_g-1}{2}}+s+1,3^{s-1}+\varepsilon_f\varepsilon_g(-3)^{\frac{s+k_f+k_g-1}{2}},3]$ code 
  and it is at least almost optimal with respect to the sphere-packing bound given in Equation (\ref{eq.Sphere-packing}).
\end{enumerate}
\end{theorem}
\begin{proof}
  The proof is very similar to that of Theorem \ref{Th. LCD even ternary} and the main difference is that we use Theorems \ref{Th. double even ternary}, \ref{Th. double odd ternary},  \ref{Th.double dual}
  and Lemma \ref{lm.jg.LCD} here.  
\end{proof}

\begin{theorem}\label{Th. LCD simple even ternary}
  Let $s=n+m$, where $n$ and $m$ are two positive integers.
  Let $g\in \mathcal{WRP}$ with $\widetilde{\mathcal{R}}_g(\beta)=\varepsilon_g{\sqrt{3^*}}^{m+k_g}\zeta_3^{g^*(\beta)}$
for $\beta \in \widetilde{\mathcal{S}\mathcal{R}}_g$, 
where $\varepsilon_g\in \{ -1, 1\}$ and $0\leq k_g\leq m$.    
Let codes $\overline{\mathcal{C}_{D_{g}(0)}}$ be defined in Equation (\ref{AC_D})
and $G$ be the generator matrix of $\overline{\mathcal{C}_{D_{g}(0)}}$.
Let $\lambda \in \F_3^*$,  then we have the following.
  \begin{enumerate}
    \item [\rm (1)] If $m+k_g$ be even, then the matrix $\overline{G}=(I\ G)$ generates a ternary LCD $[3^{s-1}+s+1,s+1,d]$ code $\mathcal{C}$,
    where $d\geq min\left\{1+2\cdot3^{s-2}-2\cdot3^{n-2}\varepsilon_g(-3)^{\frac{m+k_g}{2}},
    1+2\cdot3^{s-2}+3^{n-2}\varepsilon_g(-3)^{\frac{m+k_g}{2}}\right\}$.
    Besides, $\mathcal{C}^{\bot}$ is a ternary LCD $[3^{s-1}+s+1,3^{s-1},3]$ code
    which is at least almost optimal with respect to the sphere-packing bound given in Equation (\ref{eq.Sphere-packing}). 
    \item [\rm (2)] If $m+k_g$ be odd, then the matrix $\overline{G}=(I\ G)$ generates a ternary LCD $[3^{s-1}+s+1,s+1,d]$ code $\mathcal{C}$,
    where $d \geq min\left\{1+2\cdot3^{s-2}-3^{n-2}\varepsilon_g(-3)^{\frac{m+k_g}{2}},
    1+2\cdot3^{s-2}+3^{n-2}\varepsilon_g(-3)^{\frac{m+k_g}{2}}\right\}$.
    Besides, $\mathcal{C}^{\bot}$ is a ternary LCD $[3^{s-1}+s+1,3^{s-1},3]$ code
    which is at least almost optimal with respect to the sphere-packing bound given in Equation (\ref{eq.Sphere-packing}). 
  \end{enumerate}
\end{theorem}
\begin{proof}
  The proof is similar to Theorem \ref{Th. LCD even ternary} and we obtain the desired results from Theorems \ref{Th.simple even 0,SQ,NSQ}, \ref{Th.simple odd 0,SQ,NSQ}, \ref{Th.simple dual}
  and Lemma \ref{lm.jg.LCD}.  
\end{proof}

\begin{example}\label{LCD 3+0+-1,3+1+-1, odd -1}
  Let $G$ denote the generator matrix of code $\overline{\mathcal{C}_{D_{fg}(0)}}$ in Example \ref{3+1+-1,3+0+-1, odd -1}. 
  According to Theorem \ref{Th. LCD even ternary},  $(I\ G)$ generates a ternary LCD $[223,7,127]$ code $\mathcal{C}$ 
  and $\mathcal{C}^\bot$ is a ternary LCD $[223,216,3]$ code. Note also that $\mathcal{C}^\bot$ is optimal according to \cite{Codetables}.  
  Verified by Magma \cite{Magma}, these results are true. 
\end{example}

\begin{example}\label{LCD 1,4+0+-1, even}
  Let $G$ denote the generator matrix of code $\overline{\mathcal{C}_{D_{g}(0)}}$ in Example \ref{1,4+0+-1, even}. 
  According to Theorem \ref{Th. LCD simple even ternary}, $(I\ G)$ generates a ternary LCD [87,6,52] code $\mathcal{C}$ 
  and $\mathcal{C}^\bot$ is a optimal ternary LCD [87,81,3] code.
  Note also that $\mathcal{C}^\bot$ is optimal according to \cite{Codetables}.
  Verified by Magma \cite{Magma}, these results are true.  
\end{example}

\begin{example}\label{LCD 1,3+0+-1, odd}
  Let $G$ denote the generator matrix of code $\overline{\mathcal{C}_{D_{g}(0)}}$ in Example \ref{1,3+0+-1, odd}.
  According to Theorem \ref{Th. LCD simple even ternary}, $(I\ G)$ generates a ternary [32,5,16] LCD code $\mathcal{C}$ 
  and $\mathcal{C}^\bot$ is a optimal ternary LCD [32,27,3] code.
  Note also that $\mathcal{C}^\bot$ is optimal according to \cite{Codetables}.
  Verified by Magma \cite{Magma}, these results are true. 
\end{example}

\section{Conclusions}\label{conclude} 
In this paper, we constructed several new infinite families of ternary self-orthogonal augmented codes with flexible parameters from weakly regular plateaued functions 
and we completely determine their parameters and weight distributions as well as the parameters of their dual codes. 
It is also worth noting that these families of ternary self-orthogonal codes contain some (almost) optimal codes. 
As an application, we further derive several new infinite families of ternary LCD codes from these self-orthogonal codes 
and some of them are at least almost optimal according to the Sphere-packing bound given in Equation (\ref{eq.Sphere-packing}). 
%{\color{red}Compared with the codes presented in \cite{Ternary & bent functions}, the codes in this paper have more flexible parameters. } 
For future research, it would be interesting to construct more infinite families of $p$-ary self-orthogonal codes with good parameters by using other constructions and functions.

\section*{Acknowledgments} 
This research is supported by the National Natural Science Foundation of China under Grant No. U21A20428 and 12171134. \\

\noindent\textbf{Data availability} No data.  \\ 

\noindent\textbf{Conflict of Interest} The authors declare that there is no possible conflict of interest.

\section*{Acknowledgments} 
This research is supported by the National Natural Science Foundation of China under Grant No. U21A20428 and 12171134. 

\section*{}

\end{sloppypar}

\begin{thebibliography}{99}
%  \addtolength{\itemsep}{-1.5 em} % ��С�ο����׼�Ĵ�ֱ���
% \setlength{\itemsep}{-5pt}
 
  \bibitem{t-design} C. Bachoc, P. Gaborit, Designs and self-dual codes with long shadows, J. Combin. Theory Ser. A 105 (1) (2004) 15-34.

  \bibitem{Magma} W. Bosma, J. Cannon, C. Playoust, The Magma algebra system I: The user language. J. Symbolic Comput.  24 (3-4) (1997) 235-265.

  \bibitem{SO-40} I. Bouyukliev, S. Bouyuklieva, T.A. Gulliver, P.R.J. ${\rm \ddot{O}}$stergard, Classification of optimal binary self-orthogonal codes, J. Comb. Math. Comb. Comput. 59 (2006) 33-87.
  
  \bibitem{SIAM-SO} I. Bouyukliev, P. R. J. ${\rm \ddot{O}}$stergard, Classification of self-orthogonal codes over $\F_3$ and $\F_4$, SIAM J. Discrete Math. 19 (2) (2005) 363-370.
  
  \bibitem{quantum-1} A.R. Calderbank, E.M. Rains, P.W. Shor, N.J.A. Sloane, Quantum error correction and orthogonal geometry, Phys. Rev. Lett. 78 (1997) 405-409. 

  \bibitem{LCD code APP} C. Carlet, S. Guilley, Complementary dual codes for counter-measures to side-channel attacks, Adv. Math. Commun. 10 (1) (2016) 131-150.                                                     
   
  \bibitem{MDS-LCD-equ} C. Carlet, S. Mesnager, C. Tang, Y. Qi, Euclidean and Hermitian LCD MDS codes, Des. Codes Cryptogr. 86 (11) (2018) 2605-2618.

  \bibitem{LCD-equivalent} C. Carlet, S. Mesnager, C. Tang, Y. Qi, R. Pellikaan, Linear codes over $\F_q$ are equivalent to LCD codes for $q>3$, IEEE Trans. Inf. Theory 64 (4) (2018) 3010-3017.
 
  \bibitem{quantum-2} G. Chen, R. Li, Ternary self-orthogonal codes of dual distance three and ternary quantum codes of distance three, Des. Codes Cryptogr. 69 (2013) 53-63. 
  
  \bibitem{Cheng Y.& Cao[1]} Y. Cheng, X. Cao, Linear codes with few weights from weakly regular plateaued functions, Discrete Math. 344 (12) (2021) 112597.

  \bibitem{C-1} J.H. Conway, N.J.A. Sloane, Sphere Packing, Lattices and Groups, third ed., Springer-Verlag, New York, 1999.
 
  \bibitem{DfSeT-1} C. Ding, Linear codes from some 2-designs, IEEE Trans. Inf. Theory 60 (6) (2015) 3265-3275.

  \bibitem{secret share 3[1]} K. Ding, C. Ding, A class of two-weight and three-weight codes and their applications in secret sharing, IEEE Trans. Inf. Theory 61 (11) (2015) 5835-5842.
  
  \bibitem{DfSeT-2} C. Ding, J. Luo, H. Niederreiter, Two-weight codes punctured from irreducible cyclic codes, in: Y. Li, S. Lin, H. Niederreiter, H. Wang, C. Xing, S. Zhang (Eds.), Proceedings of the First Worshop on Coding and Cryptography, World Scientific, Singapore, 2008, pp. 119-124.
  
  \bibitem{DfSeT-3} C. Ding, H. Niederreiter, Cyclotomic linear codes of order 3, IEEE Trans. Inf. Theory 53 (6) (2007) 2274-2277.

  
  \bibitem{LCD(MDS)-1} Q. Fu, R. Li, L. Guo, LCD MDS codes from cyclic codes, Procedia Comput. Sci. 154 (2019) 663-668.
 
  \bibitem{Codetables}Grassl, M.: Bounds on the minimum distance of linear codes and quantum codes, online available at http://www.codetables.de. (Accessed on 3 August 2023).
  
  \bibitem{BentF-1} T. Helleseth, A. Kholosha, Monomial and quadratic bent functions over the finite fields of odd characteristic, IEEE Trans. Inf. Theory 52 (5) (2006) 2018-2032.

  %\bibitem{Heng-DfSeT-1} Z. Heng, F. Chen, C. Xie, D. Li, Constructions of projective linear codes by the intersection and difference of sets, Finite Fields Appl. 83 (2022) 102092.

  \bibitem{Heng-multiplicative characters} Z. Heng, C. Ding, Q. Yue, New constructions of asymptotically optimal codebooks with multiplicative characters, IEEE Trans. Inf. Theory 63 (10) (2017) 6179-6187.
  
  \bibitem{HDZ-FFA[1]} Z. Heng, C. Ding, Z. Zhou, Minimal linear codes over finite fields, Finite Fields Appl. 54 (2018) 176-196.
 
  \bibitem{Ternary & bent functions} Z. Heng, D. Li, F. Liu, Ternary self-orthogonal codes from weakly regular bent functions and their application in LCD Codes, Des. Codes Cryptogr. (2023) 1-24.

  \bibitem{DfSeT-4} Z. Heng, W. Wang, Y. Wang, Projective binary linear codes from special Boolean functions, Appl. Algebra Eng. Commun. Comput. 32 (4) (2021) 521-552.  
  
  \bibitem{Heng-DfSeT-2} Z. Heng, Q. Yue, A class of binary linear codes with at most three weights, IEEE Commun. Lett. 19 (9) (2015) 1488-1491.
  
  \bibitem{LCD(MDS)-2} X. Huang, Q. Yue, Y. Wu, X. Shi, Ternary primitive LCD BCH codes, Adv. Math. Commun. 17 (3) (2023) 644-659.
  
  \bibitem{ternary self-orthogonal} W.C. Huffman, V. Pless, {\em Fundamentals of Error-Correcting Codes}, Cambridge Univ. Press, Cambridge, 2003.

  \bibitem{IT-Jin} L. Jin, Construction of MDS codes with complementary duals, IEEE Trans. Inf. Theory 63 (5) (2017) 2843-2847.

 \bibitem{Kim-SO} J.-L. Kim, W. Choi, Self-orthogonality matrix and Reed-Muller codes, IEEE Trans. Inf. Theory  68 (11) (2022) 7159-7164.
    
 \bibitem{IEEE-kim-Lee} J.-L. Kim, Y. Kim, N. Lee, Embedding linear codes into self-orthogonal codes and their optimal minimum distances, IEEE Trans. Inf. Theory 67 (6) (2021) 3701-3707.
 
  \bibitem{LCD(MDS)-3} X. Li, F. Cheng, C. Tang, Z. Zhou, Some classes of LCD codes and self-orthogonal codes over finite fields, Adv. Math. Commun. 13 (2) (2019) 267-280.
  
  \bibitem{H-DM-SO} D. Li, Z. Heng, C. Li, Three families of self-orthogonal codes and their application in optimal quantum codes, Discrete Math. 346 (12) (2023) 113626.
  
  %\bibitem{} X. Li, Z. Heng, Constructions of near MDS codes which are optimal locally recoverable codes, Finite Fields Appl. 88 (2023) 102184.
  
  \bibitem{LCD(MDS)-4 LiS} S. Li, C. Li, C. Ding, H. Liu, Two families of LCD BCH codes, IEEE Trans. Inf. Theory 63 (9) (2017) 5699-5717.

  \bibitem{LCD(MDS)-5} S. Li, M. Shi, J. Wang, An improved method for constructing formally self-dual codes with small hulls, Des. Codes Cryptogr. 91 (2023) 2563-2583.

  \bibitem{Li-Xu-Zhao} R. Li, Z. Xu, X. Zhao, On the classification of binary optimal self-orthogonal codes, IEEE Trans. Inf. Theory 54 (8) (2008) 3778-3782.

  \bibitem{LCD(MDS)-6 liY} Y. Li, S. Zhu, E. Mart\'inez-Moro, The hull of two classical propagation rules and their applications, IEEE Trans. Inf. Theory 69 (10) (2023) 6500-6511.

  \bibitem{quadratic character[2]}  R. Lidl, H. Niederreiter, {\em Finite Fields}, Cambridge Univ. Press, Cambridge, 1997.
 
  \bibitem{quantum-3} S. Ling, J. Luo, C. Xing, Generalization of Steane's enlargement construction of quantum codes and applications, IEEE Trans. Inf. Theory 56 (8) (2010) 4080-4084.
  
  \bibitem{quantum-4} J. Liu, Ternary quantum codes of minimum distance three, Int. J. Quantum Inf. 8 (7) (2010) 1179-1186.

  \bibitem{LCD-Massey} J. Massey, Linear codes with complementary duals, Discrete Math. 106-107 (1992) 337-342.
  
  \bibitem{jg.LCD} J. Massey, Orthogonal, antiorthogonal and self-orthogonal matrices and their codes, Communications and Coding, Somerset, England: Research Studies Press, (1998) 3-9.
  
  %\bibitem{Sihem 2019[2]} S. Mesnager, F. \"Ozbudak, A. Sinak, Linear codes from weakly regular plateaued functions and their secret sharing schemes, Des. Codes Cryptogr. 87 (2-3) (2019) 463-480. 

  \bibitem{WRP31[2]} S. Mesnager, A. Sinak, Several classes of minimal linear codes with few weights from weakly regular plateaued functions, IEEE Trans. Inf. Theory 66 (4) (2020) 2296-2310. 
  
  %\bibitem{XGK2022-22[1]} M. Shi, Y. Guan, P. Sol\'e, Two new families of two-weight codes, IEEE Trans. Inf. Theory 63 (10) (2017) 6240-6246. 
 
  \bibitem{SO-DM-Vpless} V. Pless, A classification of self-orthogonal codes over GF(2), Discrete Math. 3 (1972) 209-246.
  
  \bibitem{SO-11} E.M. Rains, N.J.A. Sloane, Self-dual codes, in: {\em Handbook of Coding Theory}, Elsevier, Amsterdam, 1998, pp. 177-294.
  
  \bibitem{LCD-is-good} N. Sendrier, Linear codes with complementary duals meet the Gilbert-Varshamov bound, Discrete Math. 285 (1) (2004) 345-347.

  \bibitem{conjectures} M. Shi, S. Li, J.-L. Kim, Two conjectures on the largest minimum
  distances of binary self-orthogonal codes with dimension 5, IEEE Trans. Inf. Theory 69 (7) (2023) 4507-4512.
  
  \bibitem{Shi-book} M. Shi, Y.J. Choie, A. Sharma, P. Sol\'e, {\em Codes and Modular Froms: A Dictionary}, first edition, World Scientific, 2019.

  \bibitem{LCD(MDS)-7}  X. Shi, Q. Yue, S. Yang, New LCD MDS codes constructed from generalized Reed-Solomon codes, J. Algebra Appl. 18 (8) (2019) 1950150.
 
  \bibitem{Exponential sums[3]} A. Sinak, Construction of minimal linear codes with few weights from weakly regular plateaued functions, Turkish Journal of Math. 46 (3) (2022) 953-972.
  
  \bibitem{quantum-5}  A.M. Steane, Enlargement of Calderbank-Shor-Steane quantum codes, IEEE Trans. Inf. Theory 45 (7) (1999) 2492-2495.

  \bibitem{RF[1]} C. Tang, N. Li, F. Qi, Z. Zhou, H. Tor, Linear codes with two or three weights from weakly regular bent functions, IEEE Trans. Inf. Theory 62 (3) (2016) 1166-1176. 
  
  \bibitem{SqF[2]} C. Tang, Y. Qi, D. Huang, Two-weight and three-weight linear codes from square functions, IEEE Commun. Lett. 20 (1) (2016) 29-32.
  
  \bibitem{WRB[1]} Y. Wu, N. Li, X. Zeng, Linear codes with few weights from cyclotomic classes and weakly regular bent functions, Des. Codes Cryptogr. 88 (2020) 1255-1272.
  
  \bibitem{sym-SO} H. Xu, W. Du, On some binary symplectic self-orthogonal codes, Appl. Algebra Eng. Commun. Comput. 33 (2022) 321-337.
  
  \bibitem{define WRF[2]}Y. Zheng, X. Zhang, Plateaued functions, in: ICICS, Vol. 99, Springer, 1999, pp. 284-300.

  %\bibitem{XGK2022-35[1]} Z. Zhou, N. Li, C. Fan, T. Helleseth, Linear codes with two or three weights from quadratic bent functions, Des. Codes Cryptogr. 81 (2) (2016) 283-295. 
  
\end{thebibliography}
\end{document}